\documentclass[lettersize,journal]{IEEEtran}
\usepackage{amsmath,amsfonts}
\usepackage{algorithmic}
\usepackage{algorithm}
\usepackage{xcolor}
\usepackage{array}
\usepackage[caption=false,font=normalsize,labelfont=sf,textfont=sf]{subfig}
\usepackage{amssymb}
\usepackage{textcomp}
\usepackage{stfloats}
\usepackage{url}
\usepackage{verbatim}
\usepackage{graphicx}
\usepackage{cite}
\usepackage{cleveref}
\usepackage{amsthm} 

\newtheorem{lemma}{Lemma}
\newtheorem{theorem}{Theorem}

\newtheorem{assumption}{Assumption}

\begin{document}

\title{Context-Aware Search and Retrieval Under Token Erasure
\thanks{This work was in part supported by US NSF grants 2107370 and 2201824.}
 \thanks{This paper was presented in part at IEEE Information Theory Workshop (ITW),  Sydney, Australia, October 2025 [21].}
}
\author{Sara Ghasvarianjahromi, Joshua Barr, Yauhen Yakimenka, and J\"org Kliewer\\
\IEEEauthorblockA{Helen and John C. Hartmann Department of Electrical and Computer Engineering\\ New Jersey Institute of Technology, Newark, New Jersey, 07102, USA
\\Email: \{sg273, jb794, yauhen.yakimenka, jkliewer\}@njit.edu
}
\thanks{}
\thanks{}}



\maketitle

\begin{abstract}
 This paper introduces and analyzes a search and retrieval model for RAG-like systems under {token} erasures. We provide an information-theoretic analysis of remote document retrieval when query representations are only partially preserved. The query is represented using term-frequency-based features, and semantically adaptive redundancy is assigned according to feature importance. Retrieval is performed using TF-IDF-weighted similarity.
We characterize the retrieval error probability by showing that the vector of similarity margins converges to a multivariate Gaussian distribution, yielding an explicit approximation and computable upper bounds. Numerical results support the analysis, while a separate data-driven evaluation using embedding-based retrieval on real-world data shows that the same importance-aware redundancy principles extend to modern retrieval pipelines. Overall, the results show that assigning higher redundancy to semantically important query features improves retrieval reliability.

\end{abstract}
\begin{IEEEkeywords}
Information retrieval, retrieval-augmented generation, semantic communication, importance-aware encoding.
\end{IEEEkeywords}
\section{Introduction}
\label{sec: Introduction}
Search and retrieval systems play an important role in modern information processing, with applications ranging from web search engines to question answering and recommendation systems \cite{buttcher2016information,siriwardhana2023improving}. In particular, retrieval-augmented generation (RAG) architectures use feature-based retrieval components to select documents that provide relevant context for downstream processing \cite{dimitrakis2020survey,zhu2021retrieving,lan2021semantic,gao2023retrieval,salemi2024evaluating}. These components typically represent queries and documents using sparse or dense feature vectors, such as term frequency--inverse document frequency (TF-IDF), Best Matching 25 (BM25), or learned embeddings, and determine relevance through similarity measures including cosine similarity or $\ell_2$ distance \cite{li2024enhancing,knollmeyer2025hybrid}.

In this paper, we provide an information-theoretic characterization of retrieval error in RAG-like systems under token erasures. The results establish a theoretical foundation for importance-aware redundancy allocation in retrieval systems and clarify how semantic importance should be quantified when retrieval reliability is the primary objective. In particular, we analyze how random token erasures affect similarity-based retrieval decisions and the probability of selecting an incorrect document.

This problem is closely related to semantic communication, which emphasizes preserving task-relevant meaning rather than exact data fidelity \cite{shi2021semantic,luo2022semantic,lu2023semantics}. Early works in this area focused on representing information in a more semantically efficient manner \cite{lan2021semantic}, while more recent approaches have leveraged deep learning and generative models to develop end-to-end semantic processing frameworks \cite{yang2022semantic,chaccour2024less}. A notable direction in this literature is importance-aware semantic representation guided by pre-trained language models, where semantic relevance is used to prioritize the most informative content \cite{guo2023semantic}. A growing body of work has also investigated token-level semantic representations, in which tokens serve as fundamental semantic units and contextual modeling is used to improve efficiency and robustness \cite{qiao2025token,liu2025text}.

These ideas are particularly relevant to retrieval systems used in RAG pipelines, where query and document features must be conveyed reliably enough to preserve retrieval accuracy. However, while recent work has primarily focused on evaluating retrieval performance in RAG-style systems \cite{salemi2024evaluating,xian2024vector,lin2025gosling,wang2025balancing}, the retrieval component itself has rarely been analyzed from an information-theoretic perspective. In particular, the effect of random token erasures on similarity-based retrieval decisions and the resulting retrieval error probability has not been systematically characterized.

In earlier work \cite{ghasvarianjahromi2025context}, we took a first step toward addressing this gap by introducing a simplified analytical model for retrieval under token erasures. That work focused on a minimal setting with two candidate documents, TF-IDF-based query features, repetition-based redundancy, and squared $\ell_2$ similarity. By modeling the distortion induced by token erasures, we showed that the original and observed similarity margins converge jointly to a Gaussian distribution, enabling a closed-form approximation of the retrieval error probability. The analysis demonstrated that allocating higher redundancy to more important query features significantly reduces the probability of retrieving the wrong document.

This paper substantially extends that preliminary study. Most importantly, the analysis here applies to an arbitrary number of documents. This generalization is nontrivial, as it requires characterizing the joint distribution of multiple similarity margins and controlling the resulting high-dimensional error events. To this end, we develop a multivariate Gaussian approximation of the similarity-margin vector and derive computable upper bounds on the retrieval error probability using Bonferroni-type expansions and the \v{S}id\'ak inequality. These results provide a principled understanding of how feature importance, redundancy allocation, and query support size interact in large-scale retrieval settings.

In addition to the theoretical analysis, we complement the TF-IDF-based model with a data-driven evaluation using embedding-based retrieval on real-world data. While the analytical results are derived under a TF-IDF abstraction, the experiments demonstrate that the same importance-aware transmission principles extend to modern embedding-based retrieval pipelines. This separation clarifies the scope of the theory while illustrating the broader relevance of the proposed design approach.

The rest of this paper is organized as follows. In Section~\ref{sec:problem} we introduce the notation and terminology, the basic definitions and describe the overall system model. In Sections~\ref{sec:tfidf_analysis} and \ref{sec:embedding_analysis} we provide a detailed modeling and analysis of the individual system components specified for TF-IDF-based and embedding-based settings, respectively. Section~\ref{sec: error_pr_MVN} derives the retrieval error probability using a multivariate Gaussian approximation of the score-margin vector, establishes computable upper bounds using Bonferroni-type expansions and the \v{S}id\'ak inequality, and includes a complexity comparison of the resulting approximations in Subsection~\ref{subsec:complexity}. Section~\ref{sec: numerical_results} presents numerical results for the TF-IDF-based model together with a separate data-driven evaluation using an embedding-based retrieval pipeline. Finally, the paper is concluded in Section~\ref{sec: conclusion}.
\vspace{-6pt}
\section{Problem Statement}
\label{sec:problem}
\subsection{Notation \& Terminology}
Bold lowercase letters denote vectors, and $(\cdot)^\top$ denotes transpose. For a vector $\mathbf a$, we write $\operatorname{supp}(\mathbf a)$ for the set of indices of its nonzero entries, and $|\operatorname{supp}(\mathbf a)|$ for the cardinality of this support set. We denote by $\|\mathbf a\|_2 := \sqrt{\mathbf a^\top \mathbf a}$ the Euclidean norm. 
Finally, $\circ$ denotes the Hadamard product.

To maintain consistent terminology across the two retrieval settings studied in this paper, we use the
term \emph{token} to denote a basic transmitted query unit. 
In the TF-IDF-based setting, this unit corresponds to an indexed feature
associated with a query term, whereas in the embedding-based
setting it corresponds to a standard tokenizer-defined
text token.
Likewise, we use \emph{representation computation} as a common name for the block that produces the query-side and document-side representations used for retrieval. In the TF-IDF-based setting, this block corresponds to weighting of the recovered query features, whereas in the embedding-based setting it corresponds to encoder-based embedding generation from the surviving query tokens.
\vspace{-7pt}
\subsection{System Model}
\label{sec:system_model}

\begin{figure*}[t]
    \centering
    \includegraphics[width=\textwidth]{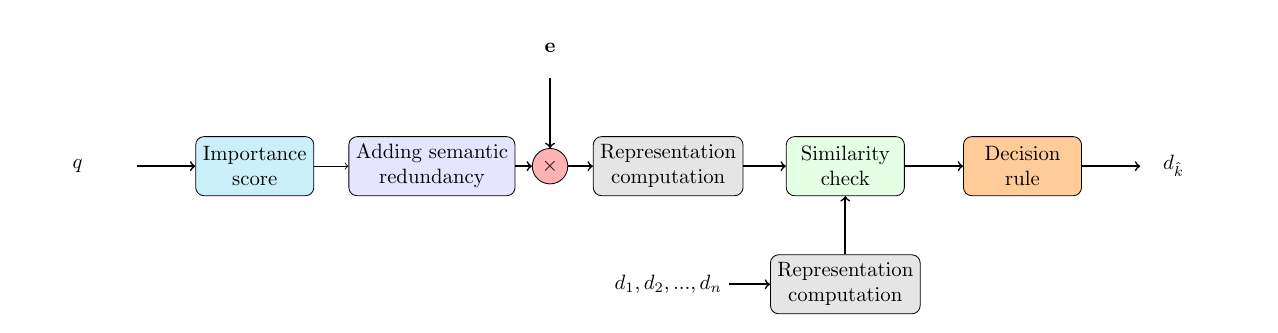}
    \caption{Remote document retrieval system. The tokenized query is first converted into a feature representation and augmented with importance-dependent redundancy. Random token erasures then affect this representation. Retrieval is performed by comparing the resulting query representation with document feature representations and selecting a document according to a predefined decision rule.}
    \label{fig:sm}
\end{figure*}
We consider a remote document retrieval system designed to identify the document most relevant to a given query $q$ consisting of $L_q$ terms. Here, ``terms'' refers to basic textual units such as words or tokens, depending on the retrieval model. The system consists of a transmitter, a receiver, and a token erasure process, as illustrated in Fig.~\ref{fig:sm}. 
Let $\mathcal{V}=\{t_1,t_2,\dots,t_N\}$ denote the vocabulary of $N$ terms from which queries and documents are formed. A query $q$ and each document $d_i$ in the corpus are formed from elements of $\mathcal{V}$.

At the transmitter, the input query $q$ is passed through an importance scoring block, which determines the relative importance of its terms. This representation may take different forms depending on the retrieval model. For example, it may correspond to a sparse term-based representation in the TF-IDF setting or to a semantic representation in the embedding-based setting. In either case, the resulting representation is intended to capture the retrieval-relevant information contained in the query and to guide the subsequent allocation of transmission redundancy.

The query information is then augmented with importance-aware redundancy to improve robustness to random erasures. Specifically, repetitions are assigned according to the importance of individual query tokens. We assume a token erasure model with probability $\epsilon$, under which some transmitted query tokens may be unavailable for retrieval. 
Such erasures may arise, for example, from packet drops caused by link outages, wireless impairments, or network congestion.

The query information available for retrieval is determined by the subset of unerased query tokens.
The quality of the resulting query representation depends on both the redundancy assigned to individual query tokens and the erasure pattern affecting them. 
The reconstructed information is then passed through an representation computation block, interpreted according to the retrieval model under consideration. In the TF-IDF case, this step combines the recovered term information with the corresponding IDF weights to form the representation used for retrieval, whereas in the embedding-based case it computes the embedding associated with the successfully received tokens. The resulting query-side representation is then used for comparison with the document representations.

In parallel, each document $d_i \in \mathcal{D}=\{d_1,d_2,\dots,d_n\}$ is converted into the representation used for retrieval by the relevant representation computation method, either TF-IDF or embeddings.
The receiver then performs a similarity check between the reconstructed query-side representation and each document representations.
Based on the resulting similarity or distance scores, a decision rule is applied to select the retrieved document, denoted by $d_{\hat{k}}$.

Thus, the overall system consists of the following common stages: importance scoring, importance-aware redundancy assignment, random token erasures affecting the query representation, representation computation for queries and documents, and retrieval through similarity comparison and decision making. This formulation is intentionally model-agnostic and serves as a common framework for the two retrieval settings studied in this paper. In the TF-IDF case, the importance scoring, redundancy allocation, and similarity measure are specified using a term-based analytical model over $\mathcal{V}$. In the data-driven case, these components are instantiated using learned embedding representations and the corresponding retrieval metric.
\vspace{-7pt}
\section{Analysis}
In this section, we specialize the general system model in Fig.~\ref{fig:sm} to two retrieval settings: the TF-IDF-based setting and the embedding-based setting, and analyze the corresponding system components in each case.
In both settings, the transmitter mitigates the effect of erasures by assigning semantic repetition budget according to the importance of the query information. The key difference lies in the object being protected and the representation used for retrieval. In the TF-IDF-based setting, erasures act on sparse TF-based feature coordinates, and retrieval is performed using a TF-IDF-weighted distance measure. In the embedding-based setting, erasures act on tokenizer-defined text tokens, and retrieval is performed using encoder-generated dense embeddings together with cosine similarity.
\vspace{-7pt}
\subsection{TF-IDF-Based Retrieval}
\label{sec:tfidf_analysis}
 We begin by characterizing the statistical structure of the vocabulary and the query generation process, which motivates the term-frequency representation adopted in this section. We then describe the TF-based importance scoring and the corresponding repetition-based redundancy assignment, where the number of repetitions is determined by term importance. Next, we model token erasures and their effect on the query representation available for retrieval. We then introduce the TF-IDF-specific representation computation used for retrieval. Finally, we formalize the TF-IDF-weighted similarity computation and the minimum-distance decision rule used for document selection, which together form the basis for the retrieval error analysis developed in the following sections.
\subsubsection{\textbf{Vocabulary and Query}}
As described in the system model, we assume a predefined vocabulary $\mathcal{V}$ with $N$ distinct terms.
For example, in a natural language processing (NLP) setting, the vocabulary may consist of all words in the English language.
Empirical studies show that term frequencies in such corpora follow Zipf’s law \cite{piantadosi2014zipf}, a heavy-tailed distribution.
According to this law, each term $t_i \in \mathcal{V}$ is assigned a unique rank $i \in \{1,2,\dots,N\}$ based on its frequency in the vocabulary, with $t_1$ being the most frequent term and $t_N$ being the least frequent.
The frequency of a term is inversely proportional to its rank $i$ and satisfies $f_i \propto \frac{1}{i^\alpha}$,
where $\alpha \geq 0$ is the Zipfian exponent controlling the decay rate.
This relationship characterizes the heavy-tailed distribution observed in natural language, where a small number of terms occur very frequently while the majority appear rarely.
To convert these frequencies into a probability distribution over $\mathcal{V}$, we introduce the normalizing constant $\sum_{j=1}^N 1/j^\alpha$.
The probability of selecting a term at rank $i$, in accordance with Zipf’s law, is then given by
\begin{equation}
    \pi_i = \frac{1/i^\alpha}{\sum_{j=1}^N 1/j^\alpha}.
    \label{eq:zipfs_distribution}
\end{equation}

To model a query $q=(w_1,w_2,\dots,w_{L_q})$, we assume that its terms are drawn independently from the vocabulary $\mathcal{V}$ according to the Zipf distribution.
The real queries exhibit term dependencies, but the i.i.d. model provides analytical tractability and acts as a first-order approximation.
This modeling choice is consistent with empirical observations that natural-language queries inherit the same heavy-tailed term-frequency structure as the underlying corpus; consequently, sampling from Zipf’s law provides a realistic model for query term occurrences and preserves the dominance of high-frequency words.
Accordingly, the vector $\mathbf{c}_q \in \mathbb{R}^N$ represents the count of each term in the query, and the frequency with which each term appears is modeled as
$\mathbf{c}_q \sim \mathrm{Multinomial}(L_q,\{\pi_1,\dots,\pi_N\})$,
where $L_q$ is query length and $\{\pi_1,\dots,\pi_N\}$ denotes the probabilities of the vocabulary terms determined by their ranks according to Zipf’s law in \eqref{eq:zipfs_distribution}.
\subsubsection{\textbf{Importance Score}}
In the TF-IDF setting, the importance of each query term is quantified through its term frequency (TF). To obtain the relative frequencies of terms in the query, we normalize the count vector $\mathbf{c}_q$ by dividing the count of each term by the query length $L_q$. This yields the normalized vector
$\widetilde{\mathbf{c}}_q = \frac{\mathbf{c}_q}{L_q}$,
which satisfies $\sum_{i=1}^N \widetilde{c}_{q,i} = 1$.
This normalization step produces a probability-like representation of the query content and is consistent with the standard TF method commonly used in natural language processing and information retrieval \cite{mishra2015analysis,ibrihich2022review}.

According to this definition, the TF value associated with term $i$ in query $q$ is given by
\begin{equation}
    \mathrm{TF}(i) = \frac{c_{q,i}}{L_q},
\end{equation}
where $c_{q,i}$ denotes the number of occurrences of term $i$ in the query $q$.
Thus, in the TF-IDF case, the term-frequency values serve as the importance scores that determine the relative contribution of the query terms and guide the subsequent allocation of transmission resources.
\subsubsection{\textbf{Adaptive Repetition Under Token Erasures}}
\label{subsec:token_erasure_model}
To improve robustness to token erasures while prioritizing the most informative parts of the query, we first derive a modified query representation $\mathbf{v}_q$ from the normalized count vector $\widetilde{\mathbf{c}}_q$. Specifically, we suppress a set of $l_s$ high-frequency but low-informative terms, commonly referred to as stop words in NLP \cite{fox1989stop}, by forcing their corresponding coordinates to zero. Let $\mathcal{L}_s \subseteq \{1,2,\dots,N\}$ denote the stop-word index set. To formalize this operation, we define a coordinate-wise masking operator $G:\mathbb{R}^N \to \mathbb{R}^N$ acting componentwise as
\begin{equation*}
G(x)_i =
\begin{cases}
0, & \text{if } i \in \mathcal{L}_s,\\[3pt]
x_i, & \text{otherwise}.
\end{cases}
\end{equation*}
The resulting stop-word-filtered query representation is
\begin{equation*}
\mathbf{v}_q = G(\widetilde{\mathbf{c}}_q).
\end{equation*}
This step removes terms that contribute little to retrieval discrimination while still consuming resources, thereby reducing overhead with limited impact on retrieval accuracy.

Although the original vector $\widetilde{\mathbf{c}}_q$ satisfies $\sum_{i=1}^N \widetilde{c}_{q,i} = 1$, the modified vector $\mathbf{v}_q$ generally satisfies $\sum_{i=1}^N v_{q,i} < 1$ because the stop-word coordinates are set to zero. Let
\begin{equation*}
\mathcal{S}_q \triangleq \mathrm{supp}(\mathbf{v}_q)=\{i:v_{q,i}\neq 0\}
\end{equation*}
and let $K_q=|\mathcal{S}_q|$ denote the number of active coordinates after stop-word removal and $M$ is the number of remaining token.

Rather than working with the full $N$-dimensional vector, we represent the query sparsely through the index--value pairs $(i,v_{q,i})$ for all $i\in\mathcal{S}_q$. This keeps only the TF-weighted coordinates that are relevant for retrieval. To make this sparse representation more robust to token erasures, we assign $r_i$ repetitions to each pair $(i,v_{q,i})$, with $r_i=0$ whenever $v_{q,i}=0$. To control the overall repetition budget, we define the design rate
\begin{equation*}
R = \frac{M}{\sum_{i=1}^N r_i},
\end{equation*}
Since coordinates with larger TF values typically play a more important role in retrieval, we assign repetitions proportionally to their TF values:
\begin{equation}
r_i = \left\lceil \frac{M}{R \sum_{j=1}^N v_{q,j}} \, v_{q,i} \right\rceil.
\label{eq:repetition}
\end{equation}
We restrict attention to queries with at least one non-stopword term, that is, $M \geq 1$. Without integer rounding, this allocation gives exactly $\sum_i r_i = M/R$. Since $\lceil x\rceil \ge x$, the rounded allocation satisfies $\sum_i r_i \ge M/R$, so the effective rate is upper bounded by $R$.

Next, we model the effect of token erasures on this repeated sparse representation. We assume that erasures occur independently across repetitions and across coordinates, for the sake of analytical traceability. If the pair $(i,v_{q,i})$ is repeated $r_i$ times, then the probability that at least one copy remains available is $1-\epsilon^{r_i}$, whereas the probability that all copies are erased is $\epsilon^{r_i}$.
To formalize this, we define an indicator vector $\mathbf{e}=[e_i]_{i=1}^N$ with independent components distributed as
\begin{equation}
e_i =
\begin{cases}
1, & \text{with probability } 1-\epsilon^{r_i}, \\[3pt]
0, & \text{with probability } \epsilon^{r_i}.
\end{cases}
\label{eq:rep_effect}
\end{equation}
Here, $e_i=1$ indicates that at least one repetition of $(i,v_{q,i})$ is retained, whereas $e_i=0$ indicates that all repetitions associated with coordinate $i$ are erased.

Accordingly, the query representation available for retrieval is obtained by preserving coordinate $i$ when $e_i=1$ and setting it to zero otherwise:
\begin{equation*}
\hat{v}_{q,i} =
\begin{cases}
v_{q,i}, & \text{if } e_i = 1, \\[3pt]
0, & \text{if } e_i = 0.
\end{cases}
\end{equation*}
Equivalently, we can write it as $\hat{\mathbf{v}}_q = \mathbf{v}_q \circ \mathbf{e}$.
 This expression makes explicit that token erasures independently preserve or remove the active coordinates of $\mathbf{v}_q$ according to the repetition-dependent probabilities in \eqref{eq:rep_effect}.
\subsubsection{\textbf{Representation Computation}}
\label{subsec:embedding_computation_tfidf}
We next specify the representation computation block in Fig.~\ref{fig:sm} for the TF-IDF retrieval setting. In this case, the block incorporates the inverse document frequency (IDF) information associated with the vocabulary terms. Specifically, let
\begin{equation}
    \mathrm{IDF}(i)
    \;=\;
    \log\!\left(\frac{n+1}{n_i+1}\right)
    \triangleq
    \zeta_i,
    \label{eq: IDF}
\end{equation}
where $n_i$ denotes the number of documents in the corpus $\mathcal{D}$ in which term $t_i$ appears. Let $\boldsymbol{\zeta}\in\mathbb{R}^N$ denote the vector of IDF weights.

The TF-IDF representation associated with the query vector available for retrieval is obtained by weighting each coordinate of $\hat{\mathbf{v}}_q$ by its corresponding IDF value. Equivalently, the query-side representation used for retrieval is $\hat{\mathbf{v}}_q \circ \boldsymbol{\zeta}$.
Thus, in the TF-IDF case, the representation computation block transforms the TF-based query representation into a TF-IDF-weighted representation that emphasizes rare and informative terms in the subsequent similarity computation.
\subsubsection{\textbf{Similarity Check and Decision Rule}}

At the receiver, the objective is to identify the document $d_{\hat{k}}$ from the corpus 
$\mathcal{D} = \{d_1, d_2, \dots, d_n\}$ that is most relevant to the transmitted query $q$.  
We assume that all unique terms appearing in the corpus are contained in the predefined vocabulary 
$\mathcal{V}$, i.e., $\bigcup_{j=1}^{n} \mathcal{T}_{d_j} \subseteq \mathcal{V}$, where $\mathcal{T}_{d_j} = \{\, t \mid t \in d_j \,\}$ denotes the set of terms in document $d_j$.

To compare the reconstructed query representation with the corpus, the receiver computes the term-frequency (TF) vectors for all documents, denoted
$\mathbf{v}_{d_1}, \mathbf{v}_{d_2}, \dots, \mathbf{v}_{d_n}$. Using the IDF weights defined in \eqref{eq: IDF}, the receiver then evaluates the TF-IDF-weighted squared $\ell_2$ distance between the reconstructed query vector $\hat{\mathbf{v}}_q$ and each document vector $\mathbf{v}_{d_j}$:
\begin{equation}
    \hat{s}_j
    \;=\;
    \sum_{i \in \mathcal{S}_q}
    \zeta_i^{2}
    \bigl(\hat{v}_{q,i} - v_{d_j,i}\bigr)^2,
    \quad j = 1,2,\dots,n.
    \label{eq: l2_distance}
\end{equation}
Only coordinates in $\mathcal{S}_q$ are evaluated, reflecting that transmission and reconstruction occur solely over these indices. We define the retrieval score over $\mathcal{S}_q$ to focus the analysis on the interaction between query features and document representations, excluding coordinates where the reconstructed query carries no information. In practice, this requires knowledge of $\mathcal{S}_q$ at the receiver, which can be ensured by communicating the query support as low-rate side information.

The document achieving the smallest value of $\hat{s}_j$ is selected as the most relevant:
\begin{equation}
    \hat{k}
    \;=\;
    \arg\min_{j \in \{1,2,\dots,n\}} \hat{s}_j.
\end{equation}
Thus, $d_{\hat{k}}$ is the document whose TF-IDF representation is closest to the reconstructed query vector under the minimum-distance decision rule. Ties, if any, are broken using a fixed tie-breaking rule, for example uniformly at random among tied candidates, so that the selected index is well-defined.
\vspace{-35pt}
\subsection{Embedding-Based Retrieval}
\label{sec:embedding_analysis}
\vspace{-7pt}
In modern NLP, pretrained encoder models such as sentence transformers \cite{reimers2019sentence} map variable-length text sequences into fixed-dimensional dense vectors, known as embeddings. These models are trained so that semantically related texts are mapped to nearby points in the embedding space, with cosine similarity serving as a standard measure of semantic proximity \cite{mikolov2013efficient}. Unlike sparse and interpretable TF-IDF representations, embeddings capture contextual and distributional meaning but are generally not directly interpretable at the coordinate level.
This embedding-based view is particularly relevant for modern retrieval pipelines, where queries and documents are represented in a dense semantic space rather than through sparse term-based features. In this section, we therefore specialize the general system model in Fig.~\ref{fig:sm} to the embedding-based retrieval setting and describe the components that differ from the TF-IDF case.
\subsubsection{\textbf{Importance Score}}

Let $q=(w_1,w_2,\dots,w_{L_q})$ denote the input query, where the terms are now processed as tokens under the tokenizer associated with the pretrained embedding model. For each query, we first compute the full-query embedding
$\mathbf{z}_{\mathrm{full}}$.
To quantify the semantic importance of token $i$, we remove that token from the query, recompute the embedding of the shortened query, and denote the result by $\mathbf{z}_{\setminus i}$.
The semantic loss induced by removing token $i$ is then measured using cosine similarity as \cite{guo2023semantic}
\begin{equation}
\mathrm{score}_i
=
1-
\frac{\mathbf{z}_{\mathrm{full}}^{\top}\mathbf{z}_{\setminus i}}
{\|\mathbf{z}_{\mathrm{full}}\|_2\,\|\mathbf{z}_{\setminus i}\|_2}.
\label{eq:semantic_importance}
\end{equation}
A large value of $\mathrm{score}_i$ indicates that removing token $i$ causes the query embedding to shift substantially, meaning that the token carries significant semantic content relative to the full query.
Repeating this procedure for all query tokens yields a semantic-importance score for each token, where larger values of $\mathrm{score}_i$ indicate greater semantic contribution to the full-query embedding.
\subsubsection{\textbf{Adaptive Repetition Under Token Erasures
}}
\label{subsec:token_erasure_embedding}
As in the TF-IDF case, redundancy is introduced to improve robustness to token erasures, but the repeated units are now query tokens rather than TF-weighted coordinates. Given the semantic-importance scores $\{\mathrm{score}_i\}_{i=1}^{L_q}$ and a total repetition budget $B=L_q/R$, we assign an integer repetition count $r_i$ to each token so as to maximize the expected preserved semantic content after erasures. Assuming independent erasures with probability $\epsilon$ for each repeated copy, the expected retained contribution of token $i$ is $(1-\epsilon^{r_i})\,\mathrm{score}_i$.

Since erasures are independent across tokens, the expected retained semantic content decomposes as a sum over tokens, which leads to the following separable optimization problem:
\begin{equation}
\begin{aligned}
\max_{r_1,\ldots,r_{L_q}} \qquad
    & \sum_{i=1}^{L_q} \bigl(1-\epsilon^{r_i}\bigr)\,\mathrm{score}_i \\[3pt]
\text{s.t.} \qquad
    & \sum_{i=1}^{L_q} r_i = B, \\[2pt]
    & r_i \in \mathbb{Z}_{\ge 0}, \qquad i=1,\ldots,L_q.
\end{aligned}
\label{eq:dp_allocation}
\end{equation}
This discrete optimization problem is solved using dynamic programming because the repetition counts are integer-valued and coupled through a total redundancy-budget constraint \cite{bellman1957dynamic}. The procedure computes the optimal value recursively over tokens and remaining budget, and then backtracks to obtain the repetition assignment. The resulting solution is importance-aware, assigning more redundancy to tokens with larger semantic contribution. In some cases, a token may be assigned $r_i=0$, meaning that it is omitted from the repeated representation. This typically occurs for tokens with negligible semantic contribution, analogous to low-informative terms in the TF-IDF case.

We next model the effect of token erasures on this repeated token representation. We assume independent erasures across repeated copies, with erasure probability $\epsilon$ per copy. A token is retained if at least one of its $r_i$ copies remains available, and is otherwise removed. Thus, the probability that token $i$ is preserved is $1-\epsilon^{r_i}$, whereas the probability that it is completely erased is $\epsilon^{r_i}$.

 In the embedding-based setting, the query representation used for retrieval is obtained by recomputing the query embedding from the surviving token sequence. We denote the resulting post-erasure query embedding by $\widehat{\mathbf{z}}$. Therefore, unlike the TF-IDF case, where erasures act directly on sparse feature coordinates, the embedding-based pipeline forms the query representation used for retrieval by re-embedding the surviving query tokens.
\subsubsection{\textbf{Similarity Check and Decision Rule}}

Let $\mathbf{z}_{d_j}$ denote the embedding of document $d_j \in \mathcal{D}$ obtained from the same pretrained encoder. Retrieval is performed by comparing the reconstructed query embedding $\widehat{\mathbf{z}}$ with the document embeddings using cosine similarity. Specifically, the score assigned to document $d_j$ is
\begin{equation}
\hat{s}_j
= {1 -}
\frac{\widehat{\mathbf{z}}^{\top}\mathbf{z}_{d_j}}
{\|\widehat{\mathbf{z}}\|_2\,\|\mathbf{z}_{d_j}\|_2},
\qquad j=1,2,\dots,n.
\label{eq:embedding_similarity}
\end{equation}
The {document with the smallest $\hat s_j$ is} retrieved:
\begin{equation}
\hat{k}
=
\arg\min_{j\in\{1,2,\dots,n\}} \hat{s}_j.
\end{equation}
To define the reference relevant document, we use the ranking induced by the full-query embedding $\mathbf{z}_{\mathrm{full}}$. In other words, the ground-truth document is taken to be the one that would be retrieved in the absence of channel erasures.

\vspace{-7pt}
\section{Theoretical Results}
\label{sec: error_pr_MVN}
In this section, we develop theoretical results for the \emph{TF-IDF-based retrieval framework} introduced in Section~\ref{sec:tfidf_analysis}. Specifically, we characterize the probability of retrieval error for remote document retrieval over an erasure channel, defined as the probability that the receiver selects a document different from the ground-truth choice obtained under perfect (no-erasure) conditions. Our analysis proceeds in three steps. First, conditioned on the query representation $\mathbf v_q$, we establish a multivariate Gaussian approximation for the TF-IDF score-margin vector induced by the erasure channel. Second, we use this approximation to characterize the resulting error probability through a Gaussian orthant probability. Third, we derive computable analytic approximations and bounds based on classical tools for Gaussian orthant probabilities, in particular Bonferroni-type expansions \cite{feller1991introduction} and the \v{S}id\'ak inequality \cite{vsidak1967rectangular}.
\vspace{-7pt}
\subsection{Preliminaries}
\label{subsec:preliminaries}
In this subsection, we introduce the score-margin representation, its coordinatewise decomposition, and the associated first- and second-order moments that will be used in the probability-of-error analysis.

For each competing document $j \neq k$, where 
$k$ denotes the ground-truth document, define the score margin
\begin{equation}
\hat{\Delta}_j
\triangleq
\hat{s}_j - \hat{s}_k ,
\label{eq:margin_j}
\end{equation}
where $\hat{s}_k$ denotes the reconstructed score of the ground-truth document, and $\hat{s}_j$ denotes the reconstructed score of competitor $j$. 
Expanding the score margin using \eqref{eq: l2_distance} and substituting $\hat v_{q,i}=e_i v_{q,i}$, we obtain
\begin{equation}
\hat{\Delta}_j
=
\sum_{i\in \mathcal S_q}
\left[
\zeta_i^{2}\bigl[v_{d_j,i}^2-v_{d_k,i}^2\bigr]
+
e_i\,2\zeta_i^2 v_{q,i}\bigl(v_{d_k,i}-v_{d_j,i}\bigr)
\right],
\label{eq:Delta_affine}
\end{equation}
where $\{e_i\}_{i\in \mathcal S_q}$ are independent Bernoulli random variables.
Thus, for each $j\neq k$, we may rewrite $\hat{\Delta}_j$ as distinct weighted sums of the same underlying erasure indicators $\{e_i\}_{i\in \mathcal S_q}$,
\begin{equation}
\hat{\Delta}_j
=
\sum_{i\in \mathcal S_q}\bigl( C_{j,i}+e_i\delta_{j,i}\bigr),
\end{equation}
where
\[
C_{j,i}
=
\zeta_i^{2}\bigl(v_{d_j,i}^2-v_{d_k,i}^2\bigr),
\quad
\delta_{j,i}
=
2\zeta_i^2 v_{q,i}\bigl(v_{d_k,i}-v_{d_j,i}\bigr).
\]

Stacking the $m=n-1$ competing margins into a vector yields
\begin{equation}
\widehat{\boldsymbol{\Delta}}
=
\bigl(\hat{\Delta}_j\bigr)_{j\neq k}
\in
\mathbb{R}^{m},
\end{equation}
which can be written as a sum of independent, though not identically distributed, random vectors indexed by the query coordinates
\begin{equation}
\widehat{\boldsymbol{\Delta}}
=
\sum_{i\in \mathcal S_q}
\bigl(\mathbf C_i + e_i\,\boldsymbol{\delta}_i\bigr),
\label{eq:Delta_C_delta_vec}
\end{equation}
where
\begin{equation}
\label{eq:delta_i}
   \mathbf C_i := (C_{j,i})_{j\neq k}, \quad \boldsymbol{\delta}_i := (\delta_{j,i})_{j\neq k},
\end{equation}
with \(\mathbf C_i\) denoting the deterministic baseline contributions of coordinate \(i\) to the competing score margins, and \(\boldsymbol{\delta}_i\) the corresponding transmission-dependent contributuin.
Defining $\mathbf U_i := \mathbf C_i + e_i\,\boldsymbol{\delta}_i$,
we obtain
\begin{equation}
\widehat{\boldsymbol{\Delta}}
=
\sum_{i\in \mathcal S_q} \mathbf U_i.
\label{eq:margin_vec}
\end{equation}

From \eqref{eq:Delta_C_delta_vec} and the facts that $\mathbb{E}[e_i]=p_i$ and $\mathrm{Var}(e_i)=p_i(1-p_i)$, it is easy to see that the conditional mean and covariance matrix of $\widehat{\boldsymbol{\Delta}}$ given $\mathbf v_q$ are
\begin{equation}
\boldsymbol{\mu}(\mathbf v_q)
=
\mathbb{E}[\widehat{\boldsymbol{\Delta}}\mid \mathbf v_q]
=
\sum_{i\in \mathcal S_q}
\bigl(\mathbf C_i + p_i\,\boldsymbol{\delta}_i\bigr),
\label{eq:mean}
\end{equation}
and
\begin{equation}
\boldsymbol{\Sigma}(\mathbf v_q)
=
\mathrm{Cov}(\widehat{\boldsymbol{\Delta}}\mid \mathbf v_q)
=
\sum_{i\in \mathcal S_q}
p_i(1-p_i)
\boldsymbol{\delta}_i
\boldsymbol{\delta}_i^\top.
\label{eq:Sigma_vq}
\end{equation}
Also, define the centered random vectors
\begin{equation}
\mathbf Y_i
:=
\mathbf U_i - \mathbb E[\mathbf U_i \mid \mathbf v_q]
=
(e_i - p_i)\,\boldsymbol\delta_i.
\label{eq:Yi_def}
\end{equation}
Then
\begin{equation}
\mathbb E[\mathbf Y_i \mid \mathbf v_q] = \mathbf 0,
\quad
\mathrm{Cov}(\mathbf Y_i \mid \mathbf v_q)
=
p_i(1-p_i)\,
\boldsymbol\delta_i
\boldsymbol\delta_i^\top .
\label{eq:Yi_cov}
\end{equation}
\begin{assumption}[Asymptotic framework]
\label{assump:asymptotic_framework}
All identities above hold for fixed vocabulary size $N$. The asymptotic statements below are interpreted along a sequence of problem instances indexed by $t$, with vocabulary sizes $N_t\to\infty$ and query representations $\mathbf v_q^{(t)}$ having supports $\mathcal S_q^{(t)} \subseteq [N_t]$ such that
$K_q^{(t)} \triangleq |\mathcal S_q^{(t)}| \to \infty$.
For each $t$, all probabilistic statements are understood conditionally on the realized query representation $\mathbf v_q^{(t)}$, and hence on the induced redundancy levels $\{r_i^{(t)}\}$ and corresponding survival probabilities $\{p_i^{(t)}\}$. For notational simplicity, we suppress the index $t$ and write $N$, $\mathcal S_q$, and $K_q = |\mathcal S_q|$, with all limits understood along this sequence.
\end{assumption}
Since \(\widehat{\boldsymbol\Delta}\) is a sum of independent but generally non-identically distributed random vectors, a standard i.i.d.\ central limit theorem does not apply directly. To establish asymptotic Gaussianity, we therefore impose the following directional Lindeberg condition on the projected summands, which ensures that in every fixed projection \(\mathbf{u}^\top\widehat{\boldsymbol\Delta}\), no single coordinate contribution \(\mathbf{u}^\top \mathbf Y_i\) carries a non-negligible fraction of the total variance \cite{van2000asymptotic}.
\begin{lemma}[Directional Lindeberg condition]
\label{lem:lindeberg}
Under Assumption~\ref{assump:asymptotic_framework}, conditioned on \(\mathbf v_q\), for every fixed unit vector \(\mathbf u\in\mathbb R^m\) and every \(\eta>0\),
\begin{equation}
\label{eq:lindeberg_directional}
\begin{aligned}
\frac{1}{\mathbf{u}^\top \boldsymbol\Sigma(\mathbf v_q) \mathbf{u}}
\sum_{i\in \mathcal S_q}
&\mathbb E\Biggl[
\bigl(\mathbf{u}^\top \mathbf Y_i\bigr)^2 \\
&\times
\mathbf 1\!\Bigl\{
|\mathbf{u}^\top \mathbf Y_i|
>
\eta\sqrt{\mathbf{u}^\top \boldsymbol\Sigma(\mathbf v_q) \mathbf{u}}
\Bigr\}
\;\Big|\; \mathbf v_q
\Biggr]
\longrightarrow 0,
\\
&\text{as } K_q\to\infty .
\end{aligned}
\end{equation}
\end{lemma}
\begin{IEEEproof}
The proof is presented in Appendix \ref{app:linbergcondition}
\end{IEEEproof}
\begin{lemma}[Asymptotic Gaussianity of the margin vector]
\label{lem:Delta_CLT}
Under Assumption~\ref{assump:asymptotic_framework} and Lemma~\ref{lem:lindeberg}, conditioned on \(\mathbf v_q\),
\begin{equation}
\widehat{\boldsymbol{\Delta}}
\xrightarrow[]{d}
\mathcal N\!\bigl(\boldsymbol{\mu}(\mathbf v_q),\boldsymbol{\Sigma}(\mathbf v_q)\bigr),
\quad \text{as } K_q\to\infty.
\end{equation}
\end{lemma}
\begin{IEEEproof}
Condition on \(\mathbf v_q\). Then the vectors \(\{\mathbf U_i\}_{i\in\mathcal S_q}\) are independent, and
\[
\widehat{\boldsymbol{\Delta}}-\boldsymbol{\mu}(\mathbf v_q)
=
\sum_{i\in\mathcal S_q}\mathbf Y_i,
\]
where \(\{\mathbf Y_i\}_{i\in\mathcal S_q}\) are independent, centered random vectors.

Fix \(\mathbf{u}\in\mathbb R^m\). Consider the scalar variables
\[
Z_i:=\mathbf{u}^\top \mathbf Y_i,
\qquad i\in\mathcal S_q.
\]
Conditioned on \(\mathbf v_q\), the variables \(\{Z_i\}_{i\in\mathcal S_q}\) are independent and centered, and their total variance is
\[
\sum_{i\in\mathcal S_q}\mathrm{Var}(Z_i\mid \mathbf v_q)
=
\sum_{i\in\mathcal S_q}\mathbf{u}^\top \mathrm{Cov}(\mathbf Y_i\mid \mathbf v_q)\mathbf{u}
=
\mathbf{u}^\top \boldsymbol{\Sigma}(\mathbf v_q)\mathbf{u}.
\]
Lemma~\ref{lem:lindeberg} is precisely the Lindeberg condition for the triangular array \(\{Z_i\}_{i\in\mathcal S_q}\). 
Moreover, by \eqref{eq:Sigma_vq},
$\mathbf{u}^\top \boldsymbol{\Sigma}(\mathbf v_q)\mathbf{u}\ge0$.
Therefore, if \(\mathbf{u}^\top \boldsymbol{\Sigma}(\mathbf v_q)\mathbf{u}>0\), the Lindeberg--Feller central limit theorem for triangular arrays implies that
\[
\frac{\sum_{i\in\mathcal S_q} Z_i}
{\sqrt{\mathbf{u}^\top \boldsymbol{\Sigma}(\mathbf v_q)\mathbf{u}}}
\xrightarrow[]{d}
\mathcal N(0,1),
\qquad \text{as } K_q\to\infty.
\]
Equivalently,
\[
\mathbf{u}^\top\bigl(\widehat{\boldsymbol{\Delta}}-\boldsymbol{\mu}(\mathbf v_q)\bigr)
\xrightarrow[]{d}
\mathcal N\!\bigl(0,\,\mathbf{u}^\top\boldsymbol{\Sigma}(\mathbf v_q)\mathbf{u}\bigr).
\]

If \(\mathbf{u}^\top \boldsymbol{\Sigma}(\mathbf v_q)\mathbf{u}=0\), then \(\sum_{i\in\mathcal S_q}\mathrm{Var}(Z_i\mid \mathbf v_q)=0\), so \(Z_i=0\) almost surely for every \(i\in\mathcal S_q\). Hence
\[
\mathbf{u}^\top\bigl(\widehat{\boldsymbol{\Delta}}-\boldsymbol{\mu}(\mathbf v_q)\bigr)=0
\quad \text{almost surely}.
\]

Thus, for every fixed \(\mathbf{u}\in\mathbb R^m\),
\[
\mathbf{u}^\top\bigl(\widehat{\boldsymbol{\Delta}}-\boldsymbol{\mu}(\mathbf v_q)\bigr)
\xrightarrow[]{d}
\mathcal N\!\bigl(0,\,\mathbf{u}^\top\boldsymbol{\Sigma}(\mathbf v_q)\mathbf{u}\bigr).
\]
Hence, by the Cramér--Wold theorem \cite{billingsley2013convergence},
\[
\widehat{\boldsymbol{\Delta}}
\xrightarrow[]{d}
\mathcal N\!\bigl(\boldsymbol{\mu}(\mathbf v_q),\boldsymbol{\Sigma}(\mathbf v_q)\bigr).
\]
\end{IEEEproof}
\begin{theorem}[Asymptotically exact Gaussian approximation of the conditional error probability]
\label{thm:Pe_MVN_smooth}
Let
\[
P_e(\mathbf v_q)=\Pr(\widehat{k}\neq k \mid \mathbf v_q)
\]
denote the conditional probability of error. Under the asymptotic normality of
\(\widehat{\boldsymbol{\Delta}}\) established in Lemma~\ref{lem:Delta_CLT}, suppose that the limiting Gaussian law
\(\mathcal N\!\bigl(\boldsymbol\mu(\mathbf v_q),\boldsymbol\Sigma(\mathbf v_q)\bigr)\)
assigns zero probability to the boundary of the positive orthant. Then
\begin{equation}
\label{eq:Pe_MVN_final_smooth}
P_e(\mathbf v_q)
-
\left[
1-\Phi_m\!\bigl(\mathbf 0;\,-\boldsymbol\mu(\mathbf v_q),\,\boldsymbol\Sigma(\mathbf v_q)\bigr)
\right]
\to 0,
\quad \text{as } K_q\to\infty.
\end{equation}
\end{theorem}
\begin{IEEEproof}
By definition of the score margins,
\[
\widehat{k}=k
\quad\Longleftrightarrow\quad
\hat{\Delta}_j>0
\quad\text{for all } j\neq k,
\]
which is equivalent to
\[
\widehat{\boldsymbol\Delta}\in\mathbb R_+^m,
\qquad
\mathbb R_+^m:=\{x\in\mathbb R^m:\ x_j>0,\ \forall j\}.
\]
Hence
\[
P_e(\mathbf v_q)
=
1-\Pr(\widehat{\boldsymbol\Delta}\in\mathbb R_+^m).
\]

Let $Z\sim\mathcal N\!\bigl(\boldsymbol\mu(\mathbf v_q),\boldsymbol\Sigma(\mathbf v_q)\bigr)$.
By Lemma~\ref{lem:Delta_CLT},
$\widehat{\boldsymbol\Delta}\xrightarrow[]{d} Z$.
Moreover, by assumption,
$\Pr(Z\in \partial \mathbb R_+^m)=0$,
so \(\mathbb R_+^m\) is a continuity set for \(Z\) \cite{billingsley2013convergence}. Therefore,
\[
\Pr(\widehat{\boldsymbol\Delta}\in\mathbb R_+^m)
\to
\Pr(Z\in\mathbb R_+^m).
\]
Since $\Pr(Z\in\mathbb R_+^m)
=
\Phi_m\!\bigl(\mathbf 0;\,-\boldsymbol\mu(\mathbf v_q),\,\boldsymbol\Sigma(\mathbf v_q)\bigr)$,
it follows that
\[
P_e(\mathbf v_q)
=
1-\Pr(\widehat{\boldsymbol\Delta}\in\mathbb R_+^m)
\to
1-\Phi_m\!\bigl(\mathbf 0;\,-\boldsymbol\mu(\mathbf v_q),\,\boldsymbol\Sigma(\mathbf v_q)\bigr),
\]
which is exactly \eqref{eq:Pe_MVN_final_smooth}.
\end{IEEEproof}
\subsection{Bonferroni-Type Bounds}
To obtain computable bounds on the conditional probability of error, stated in Theorem~\ref{thm:Pe_MVN_smooth}, define the
competitor error events
\begin{equation*}
E_j := \{\hat{\Delta}_j \le 0\}, \qquad j=1,\dots,m.
\end{equation*}
Since retrieval is incorrect if and only if at least one competitor has
nonpositive margin, we have
\begin{equation*}
P_e(\mathbf v_q)
=
\Pr\!\left(\bigcup_{j=1}^m E_j \,\middle|\, \mathbf v_q\right).
\end{equation*}
For any integer \(b\ge 1\), define the \(b\)-fold intersection sums
\begin{equation}
B_b
=
\sum_{1\le j_1<\cdots<j_b\le m}
\Pr\!\bigl(E_{j_1}\cap\cdots\cap E_{j_b}\mid \mathbf v_q\bigr).
\label{eq:bonferroni_sums}
\end{equation}
By the inclusion--exclusion principle,
\begin{equation}
P_e(\mathbf v_q)
=
\sum_{b=1}^{m}(-1)^{b+1} B_b.
\label{eq:inclusion_exclusion}
\end{equation}
The Bonferroni inequalities \cite{feller1991introduction} are obtained by
truncating the inclusion--exclusion expansion in
\eqref{eq:inclusion_exclusion}. In particular, for every positive integer
\(t\) such that \(2t-1\le m\),
\begin{equation}
\sum_{b=1}^{2t}(-1)^{b+1} B_b
\;\le\;
P_e(\mathbf v_q)
\;\le\;
\sum_{b=1}^{2t-1}(-1)^{b+1} B_b.
\label{eq:bonferroni_general}
\end{equation}
Hence, truncation after an odd number of terms yields a rigorous upper bound,
while truncation after an even number of terms yields a lower bound. In
particular, choosing \(t=1\) and \(t=2\) gives
\begin{equation*}
P_e(\mathbf v_q)\le B_1,
\qquad
P_e(\mathbf v_q)\le B_1-B_2+B_3.
\end{equation*}
\paragraph*{\textbf{First-order Bonferroni approximation}}
Under the Gaussian approximation in Lemma~\ref{lem:Delta_CLT}, each marginal event probability is approximated by the corresponding one-dimensional Gaussian tail.
Therefore, the resulting first-order Bonferroni approximation is the union bound given as
\begin{equation}
\label{eq:bonferroni_first}
B^{(1)}(\mathbf v_q)
:=
\sum_{j=1}^{m}
\Phi\!\left(
-\frac{\mu_j(\mathbf v_q)}
{\sqrt{\Sigma_{jj}(\mathbf v_q)}}
\right).
\end{equation}
\paragraph*{\textbf{Third-order Bonferroni approximation}}
Sharper approximations are obtained by incorporating pairwise and triple intersections.
Under the Gaussian approximation of Lemma~\ref{lem:Delta_CLT},
\begin{equation*}
\Pr(E_j\cap E_\ell \mid \mathbf v_q)
\approx
\Phi_2\!\bigl(
\mathbf 0;
\boldsymbol\mu_{j\ell}(\mathbf v_q),
\boldsymbol\Sigma_{j\ell}(\mathbf v_q)
\bigr),
\end{equation*}
and
\begin{equation*}
\Pr(E_j\cap E_\ell\cap E_u \mid \mathbf v_q)
\approx
\Phi_3\!\bigl(
\mathbf 0;
\boldsymbol\mu_{j\ell u}(\mathbf v_q),
\boldsymbol\Sigma_{j\ell u}(\mathbf v_q)
\bigr),
\end{equation*}
where \(\boldsymbol\mu_{j\ell}(\mathbf v_q)\) and \(\boldsymbol\Sigma_{j\ell}(\mathbf v_q)\) denote the subvector and submatrix of
\(\boldsymbol\mu(\mathbf v_q)\) and \(\boldsymbol\Sigma(\mathbf v_q)\) corresponding to indices \((j,\ell)\), and similarly
\(\boldsymbol\mu_{j\ell u}(\mathbf v_q)\) and \(\boldsymbol\Sigma_{j\ell u}(\mathbf v_q)\) correspond to indices \((j,\ell,u)\).

Accordingly, the Gaussian approximations of the second- and third-order intersection sums are
\begin{align}
\widetilde{B}_2(\mathbf v_q)
&=
\sum_{1\le j<\ell\le m}
\Phi_2\!\bigl(
\mathbf 0;
\boldsymbol\mu_{j\ell}(\mathbf v_q),
\boldsymbol\Sigma_{j\ell}(\mathbf v_q)
\bigr),
\label{eq:S2_G}\\[4pt]
\widetilde{B}_3(\mathbf v_q)
&=
\sum_{1\le j<\ell<u\le m}
\Phi_3\!\bigl(
\mathbf 0;
\boldsymbol\mu_{j\ell u}(\mathbf v_q),
\boldsymbol\Sigma_{j\ell u}(\mathbf v_q)
\bigr).
\label{eq:S3_G}
\end{align}
Combining these terms with \(B^{(1)}(\mathbf v_q)\) from \eqref{eq:bonferroni_first} yields the Gaussian third-order Bonferroni approximation
\begin{equation}
\label{eq:bonferroni_third}
B^{(3)}(\mathbf v_q)
:=
B^{(1)}(\mathbf v_q)-\widetilde{B}_2(\mathbf v_q)+\widetilde{B}_3(\mathbf v_q).
\end{equation}
Accordingly, $\min\{B^{(1)}(\mathbf v_q),\widetilde{B}^{(3)}(\mathbf v_q)\}$ may be used as a tighter computable approximation.
\vspace{-7.5pt}
\subsection{\v{S}id\'ak Bound}
Bonferroni-type bounds provide rigorous and systematically improvable upper bounds on the union probability associated with competitor error events, and hence on the conditional probability of error. However, their computational complexity grows rapidly with the order of the expansion, since higher-order approximations require evaluating an increasing number of low-dimensional Gaussian orthant probabilities. When the number of competitors \(m\) is large, this combinatorial growth limits their practical usefulness. As an alternative, we invoke the \v{S}id\'ak inequality \cite{vsidak1967rectangular}, which yields a tractable lower bound on multivariate Gaussian orthant probabilities in terms of one-dimensional marginals. When the Gaussian margin vector has nonnegative pairwise correlations, the \v{S}id\'ak inequality applies directly. In the present model, this condition is shown below to hold asymptotically as the support size increases, making the resulting \v{S}id\'ak approximation increasingly tight and computationally efficient.
For an \(m\)-dimensional Gaussian random vector \(\mathbf Z=(Z_1,\dots,Z_m)^\top\sim\mathcal N(\boldsymbol\mu,\boldsymbol\Sigma)\) whose pairwise correlation coefficients satisfy
\[
\rho_{j\ell}
=
\frac{\Sigma_{j\ell}}{\sqrt{\Sigma_{jj}\Sigma_{\ell\ell}}}
\ge 0,
\qquad j\neq \ell,
\]
the \v{S}id\'ak inequality states that \cite{vsidak1967rectangular}
\begin{equation}
\label{eq:sidak_lower}
\Pr(\mathbf Z>\mathbf 0)
\ge
\prod_{j=1}^m \Pr(Z_j>0),
\end{equation}
with equality if and only if the coordinates are independent.
Since $\Pr(Z_j>0)=\Phi\!\bigl(\mu_j/\sqrt{\Sigma_{jj}}\bigr)$ for a univariate Gaussian
random variable, the inequality in~\eqref{eq:sidak_lower} yields the explicit
product lower bound
\begin{equation}
\label{eq:sidak_explicit}
\Pr(\mathbf Z > \mathbf 0)
\;\ge\;
\prod_{j=1}^{m}
\Phi\!\left(
\frac{\mu_j}{\sqrt{\Sigma_{jj}}}
\right).
\end{equation}
Applying the \v{S}id\'ak inequality to the Gaussian approximation of the margin vector in Theorem~\ref{thm:Pe_MVN_smooth} yields the following approximation. If all pairwise correlations in \(\boldsymbol\Sigma(\mathbf v_q)\) are nonnegative, then \eqref{eq:sidak_explicit} gives
\begin{equation}
\label{eq:sidak_orthant}
\Pr(\widehat{\boldsymbol\Delta} > \mathbf 0 \mid \mathbf v_q)
\;\gtrsim\;
\prod_{j=1}^{m}
\Phi\!\left(
\frac{\mu_j(\mathbf v_q)}
{\sqrt{\Sigma_{jj}(\mathbf v_q)}}
\right).
\end{equation}
The nonnegativity condition is essential: for Gaussian vectors, \eqref{eq:sidak_explicit} does not hold in general when some pairwise correlations are negative. Thus, whenever the correlation condition is satisfied, the right-hand side of \eqref{eq:sidak_orthant} gives a rigorous lower bound on the corresponding Gaussian orthant probability. Equivalently, this leads to the following approximation for the conditional probability of error:
\begin{equation}
\label{eq:sidak_error_bound}
P_e(\mathbf v_q)
\;\lesssim\;
1-
\prod_{j=1}^{m}
\Phi\!\left(
\frac{\mu_j(\mathbf v_q)}
{\sqrt{\Sigma_{jj}(\mathbf v_q)}}
\right).
\end{equation}
Here, \(\gtrsim\) and \(\lesssim\) indicate approximations induced by the Gaussian approximation, rather than rigorous inequalities for the exact conditional error probability.

In the present setting, empirical evaluation of \(\boldsymbol\Sigma(\mathbf v_q)\) indicates that the vast majority of pairwise correlations are nonnegative, with any negative correlations observed for small query lengths typically being weak. In such cases, the product expression in \eqref{eq:sidak_orthant} may still serve as a heuristic approximation in practice, although it is no longer guaranteed to be a rigorous bound when the nonnegativity condition is violated. This expression is fully analytic, inexpensive to evaluate, and scales linearly in the number of competitors. As shown next, the nonnegativity condition holds asymptotically in the present model as the number of active query coordinates grows.

To justify the nonnegativity condition in a tractable asymptotic setting, we
introduce in the next lemma an auxiliary random-ensemble model for the document
coordinates. This i.i.d.\ assumption is used only to analyze the limiting sign
of the Gaussian margin correlations in \(\boldsymbol{\Sigma}(\mathbf v_q)\); it
is not part of the main fixed-corpus retrieval model, but rather provides a
probabilistic justification for the \v{S}id\'ak approximation in an
asymptotic regime.
\begin{lemma}[Asymptotic positivity of competitor correlations under an i.i.d.\ model]
\label{lem:cov_concentration}
Fix two distinct competitors \(j\neq \ell\). Assume that, for each \(i\in\mathcal S_q\), the random variables
\(\{v_{d_r,i}:r\in\{k,j,\ell\}\}\) are i.i.d.\ copies of a real random variable \(F\) with
\(\mathbb E[F^2]<\infty\) and \(\mathrm{Var}(F)>0\), and that these triples are independent across \(i\in\mathcal S_q\).
Assume also that the weights
\[
a_i := 4\,p_i(1-p_i)\,\zeta_i^4\,v_{q,i}^2, \qquad i\in\mathcal S_q,
\]
satisfy
\[
c_1 K_q \le \sum_{i\in\mathcal S_q} a_i \le c_2 K_q
\]
for some constants \(0<c_1\le c_2<\infty\) and all sufficiently large \(K_q\), and that
\(\sup_{i\in\mathcal S_q} a_i<\infty\).
Then
\[
\rho_{j\ell}(\mathbf v_q)
=
\frac{\Sigma_{j\ell}(\mathbf v_q)}
{\sqrt{\Sigma_{jj}(\mathbf v_q)\Sigma_{\ell\ell}(\mathbf v_q)}}
\xrightarrow[K_q\to\infty]{\mathrm{a.s.}}
\frac{1}{2}.
\]
In particular, \(\rho_{j\ell}(\mathbf v_q)>0\) for all sufficiently large \(K_q\), almost surely.
\end{lemma}
\begin{IEEEproof}
The proof is presented in Appendix \ref{app:covariance_positivity}.
\end{IEEEproof}
\vspace{-10pt}
\subsection{Complexity}
\label{subsec:complexity}
In this subsection, we compare the computational complexity of the multivariate normal (MVN)
approximation and the computable approximations developed above. All methods
share the same preliminary moment computation based on
\(\boldsymbol{\mu}(\mathbf v_q)\) and \(\boldsymbol{\Sigma}(\mathbf v_q)\) in \eqref{eq:mean} and \eqref{eq:Sigma_vq}. Therefore,
the main complexity differences mainly arise in the evaluation of the resulting
approximation or bound.
The exact MVN approximation in Theorem~1 requires
evaluation of an \(m\)-dimensional Gaussian orthant probability, where
\(m=n-1\) is the number of competitors. In practice, this typically involves an
\(\mathcal O(m^3)\) covariance factorization, together with additional
numerical cost for MVN integration or Monte Carlo evaluation. When \(m\) is
large, this cost becomes prohibitive, which motivates the use of simpler
analytic approximations.
Among the computable alternatives, the first-order Bonferroni approximation
\(B^{(1)}(\mathbf v_q)\) and the \v{S}id\'ak approximation both require only
\(m\) univariate Gaussian CDF evaluations once the relevant moments are
available, and therefore scale linearly in \(m\). By contrast, the third-order
Bonferroni approximation \(B^{(3)}(\mathbf v_q)\) requires summation over
\(\binom{m}{3}\) triplets and therefore has complexity \(\mathcal O(m^3)\).
The \v{S}id\'ak expression is the least expensive overall, since computing
\(\boldsymbol{\mu}(\mathbf v_q)\) and
\(\operatorname{diag}(\boldsymbol{\Sigma}(\mathbf v_q))\) from the
coordinatewise decomposition over the active support \(S_q\) costs
\(\mathcal O(K_q m)\), yielding total complexity \(\mathcal O(K_q m)\) while
avoiding covariance factorization and higher-dimensional Gaussian integration.

\vspace{-15pt}
\section{Numerical Results}
\label{sec: numerical_results}
In this section, we evaluate the retrieval error probability of the proposed system and assess the accuracy of the theoretical approximations and bounds.
Using synthetic data generated to match the assumptions of the theoretical model, we compare Monte Carlo simulations with the Gaussian approximation, the Bonferroni bounds, and the \v{S}id'ak bound.
We also report simulation results on real-world query--document corpora to illustrate system behavior beyond the idealized setting. The results show how query length, erasure probability, and coding rate affect retrieval performance. They demonstrate the accuracy and tightness of the derived approximations and bounds in finite-dimensional regimes under the theoretical model.
\vspace{-20pt}
\subsection{Datasets}
This section describes the datasets used in the numerical evaluation. We consider both a synthetic dataset designed to exactly match the statistical assumptions of the theoretical model and a real-world corpus, preprocessed and augmented to enable controlled, simulation-based experimentation.
\subsubsection{\textbf{Synthetic dataset}}
We generate a synthetic dataset to evaluate the retrieval error probability under the stochastic erasure model. The purpose is to construct a document corpus and queries whose statistics are consistent with the assumptions of the theoretical model.
We consider a document collection of size $n$, where each document is represented by a sparse bag-of-words vector over a vocabulary $\mathcal{V}$ of size $N$. The document terms are sampled i.i.d.\ from a Zipf distribution with exponent $\alpha = 1$, capturing the heavy-tailed term statistics commonly observed in natural language. Each document has a fixed length $L_{\mathrm{doc}}=20000$, and its term-frequency vector is formed accordingly.
In each Monte Carlo trial, a query of length $L_q$ is generated independently by sampling its terms i.i.d.\ from the same Zipf distribution.
\subsubsection{\textbf{Real dataset}}
To construct a realistic evaluation corpus, we use the Natural Questions (NQ) dataset released by Google \cite{kwiatkowski2019natural}, which contains real user queries paired with long-answer passages from Wikipedia. We retain only the question--answer pairs and discard the full articles. Each answer passage is cleaned by HTML stripping and additional preprocessing, including the removal of tables, special symbols, and non-text artifacts, resulting in a text-only corpus of $n$ queries and their corresponding documents.

To expand the query set while preserving semantic consistency, we apply an LLM-based augmentation procedure. For each original question--document pair, a Llama3 (8B) model \cite{dubey2024llama} generates 29 additional questions associated with the same answer passage. The prompt is designed to preserve semantic consistency while avoiding near-duplicate outputs, producing a naturally long-tailed distribution of query variants. This yields a corpus of $30$ queries associated with the same $n$ documents.

A qualitative embedding analysis using Uniform Manifold Approximation and Projection (UMAP)~\cite{mcinnes2018umap}, shown in Fig.~\ref{fig:UMAP}, confirms that the augmented queries remain clustered around the original queries. This indicates minimal semantic drift and supports the use of the augmented corpus for simulation-based evaluation.
\begin{figure}
    \centering
\includegraphics[width=0.32\textwidth,height=0.19\textheight]{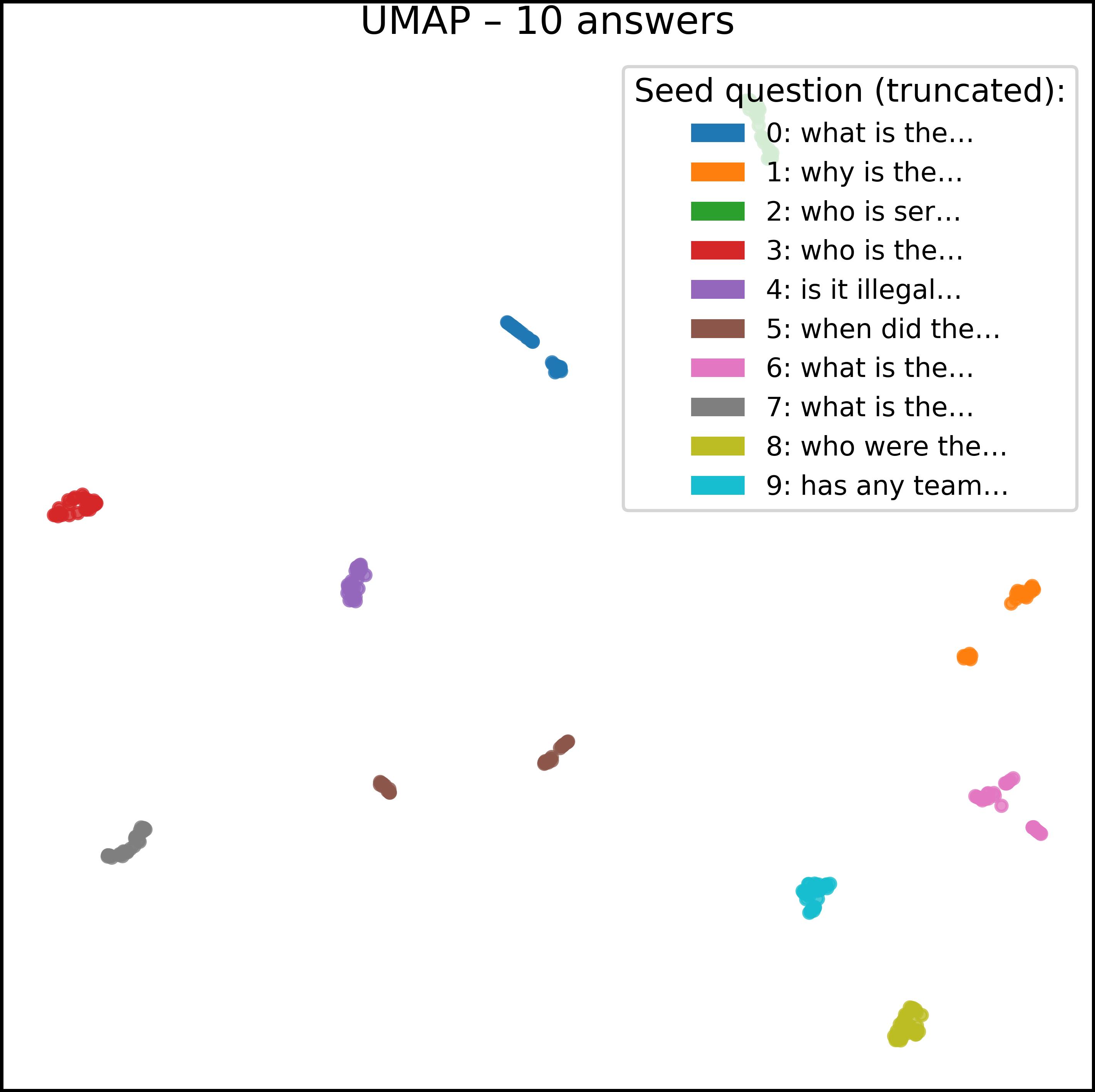}
    \caption{UMAP of generated question embeddings}
    \label{fig:UMAP}
\end{figure}
\vspace{-11pt}
\subsection{TF-IDF-Based Retrieval Results (Synthetic Data)} 
Figures~\ref{fig:error_synthetic_n_10} and~\ref{fig:error_synthetic_n_100} show the retrieval error probabilities obtained on synthetic data for $n=10$ and $n=100$, respectively. As expected, the error probability increases with the erasure probability~$\epsilon$ and decreases as the query length $L_q$ increases and the rate $R$ decreases, corresponding to higher redundancy. At \(\epsilon=1\), the error probability becomes one for all query lengths because, when retrieval is performed over the support of \(\mathbf v_q\), the receiver selects the document closest to the zero vector on that support, which is necessarily different from the ground-truth document.
Larger values of $L_q$ typically increase the number of active query coordinates $K_q$, so that the score margins aggregate a larger number of independent coordinatewise contributions. This leads to stronger concentration of the margins and improved separation between the ground-truth document and its competitors, which is consistent with the Gaussian approximation developed in \cref{sec: error_pr_MVN}.

The multivariate normal (MVN) approximation closely tracks the simulation results across most operating regimes. The largest discrepancies appear in the small-$\epsilon$ regime, where retrieval errors are rare, and the performance is governed by extreme tail events together with finite-$K_q$ effects. Among the computable analytic bounds shown in the figures, the \v{S}id\'ak bound is generally the tightest, improving on both the first- and third-order Bonferroni bounds over the plotted parameter ranges\footnote{The computable analytic bounds appear empirically tightest in the small-erasure regime, which is also the most practically relevant operating region.}.
\begin{figure*}[t]
\centering
\includegraphics[width=0.4\textwidth,height=0.17\textheight]{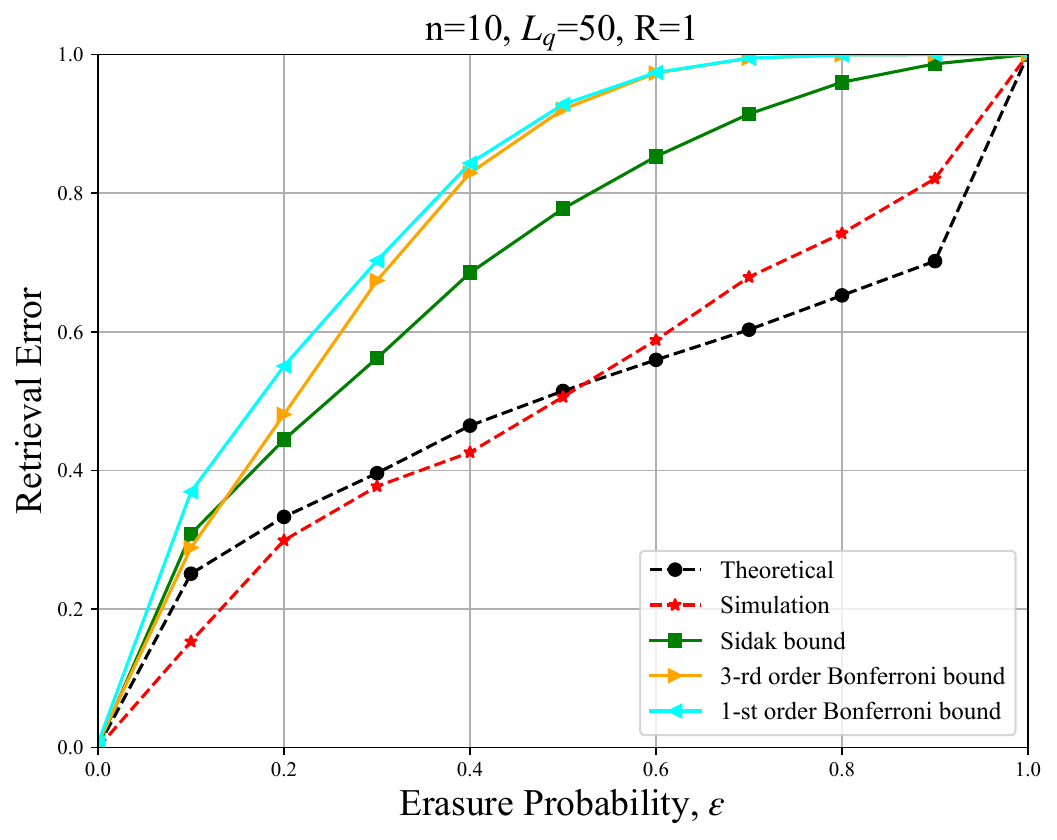}\hspace{0.02\textwidth}
\includegraphics[width=0.4\textwidth,height=0.17\textheight]{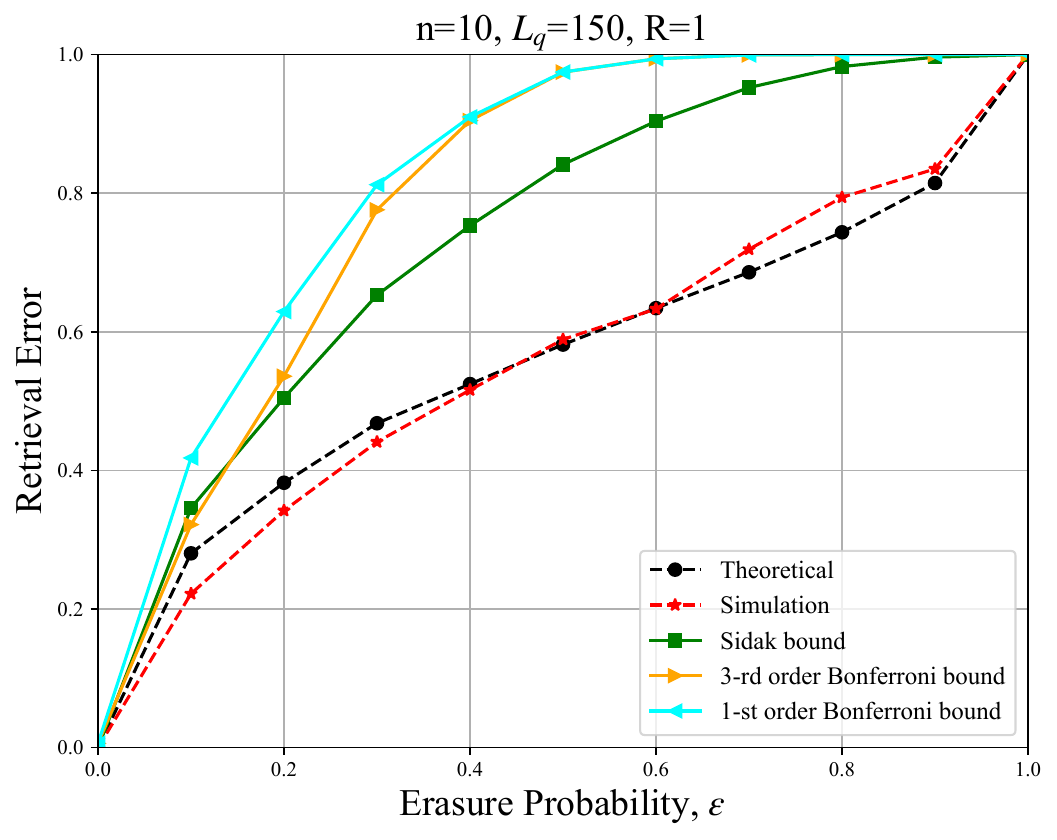}

\vspace{0.8em}

\includegraphics[width=0.4\textwidth,height=0.17\textheight]{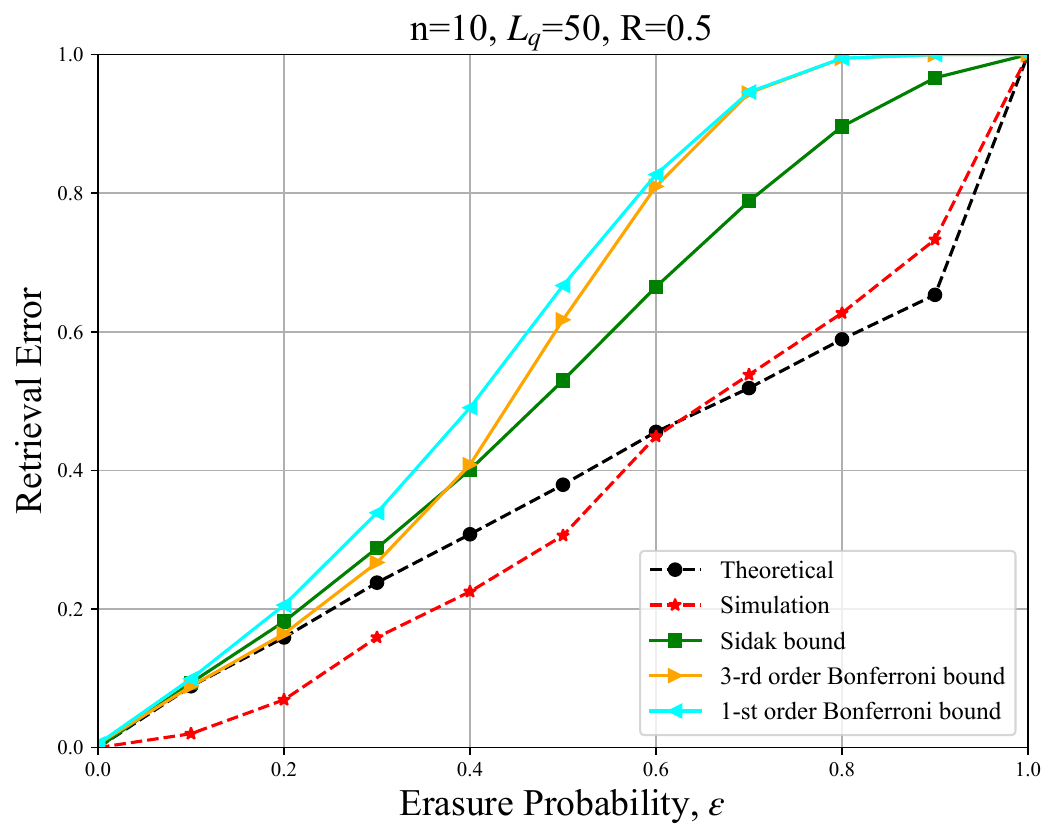}\hspace{0.02\textwidth}
\includegraphics[width=0.4\textwidth,height=0.17\textheight]{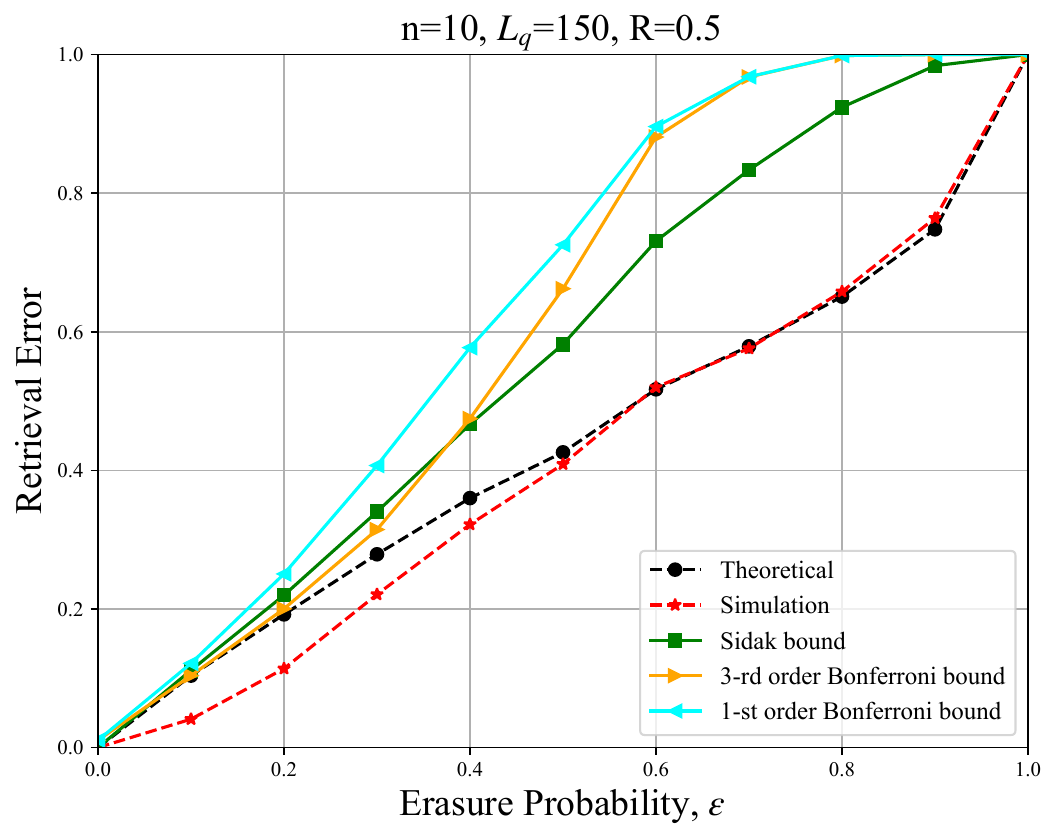}

\caption{Retrieval error probability versus erasure probability for $n=10$ documents, shown for different query lengths $L_q$ and data rates $R=1$ and $\frac{1}{2}$.}
\label{fig:error_synthetic_n_10}
\end{figure*}

\begin{figure*}[t]
\centering
\includegraphics[width=0.4\textwidth,height=0.17\textheight]{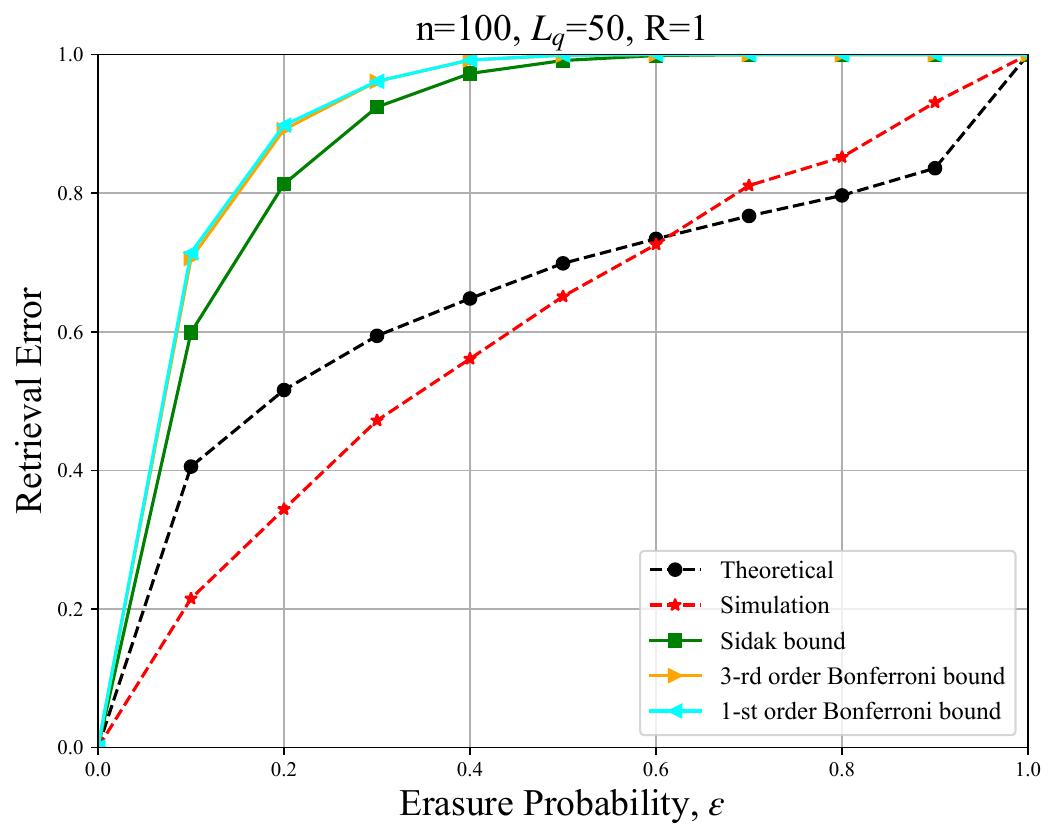}\hspace{0.02\textwidth}
\includegraphics[width=0.4\textwidth,height=0.17\textheight]{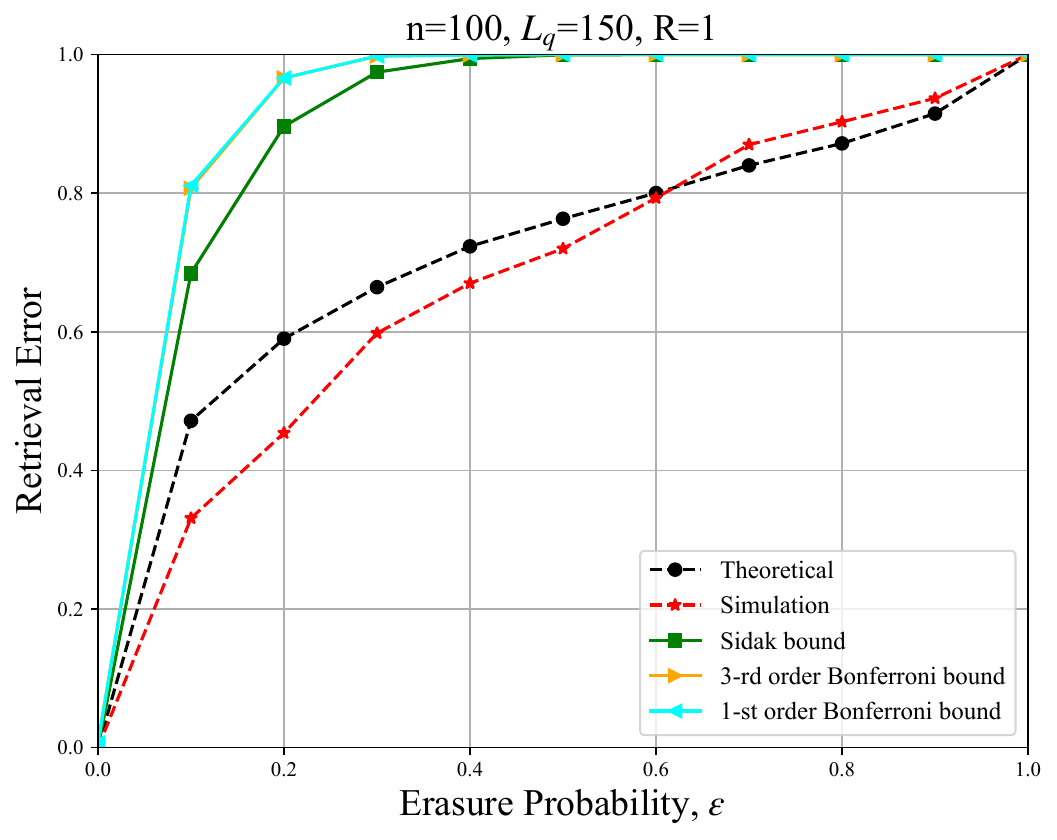}

\vspace{0.8em}

\includegraphics[width=0.4\textwidth,height=0.17\textheight]{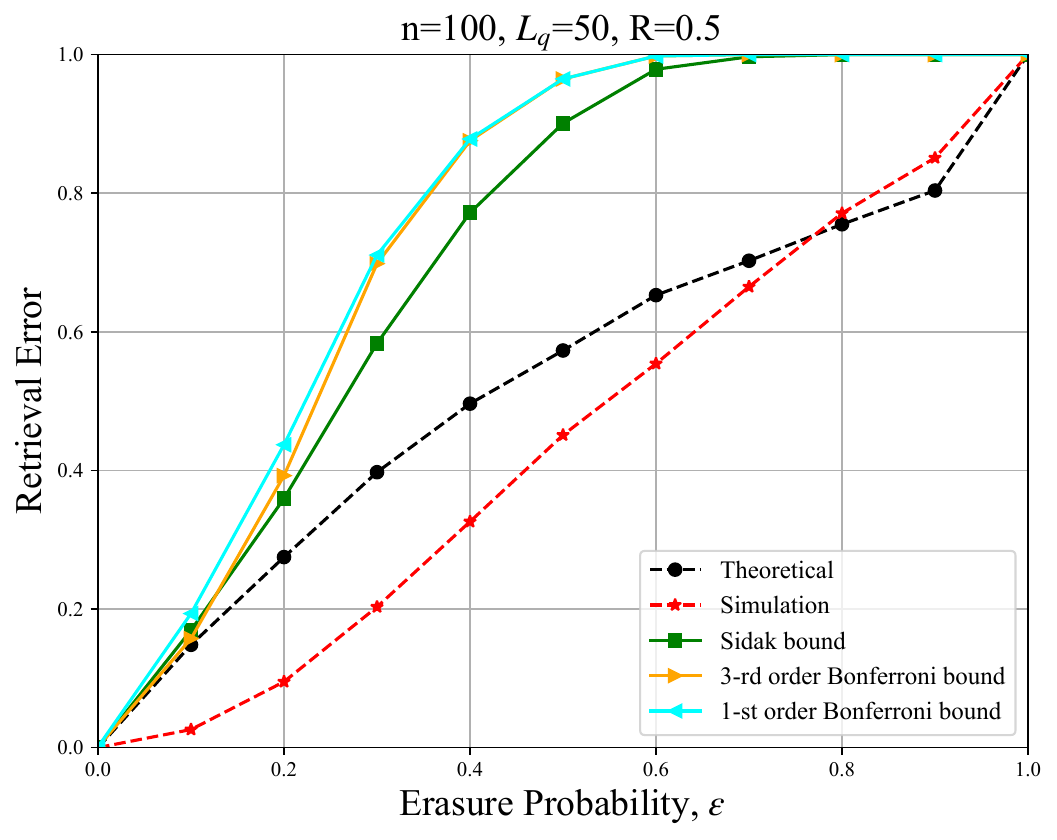}\hspace{0.02\textwidth}
\includegraphics[width=0.4\textwidth,height=0.17\textheight]{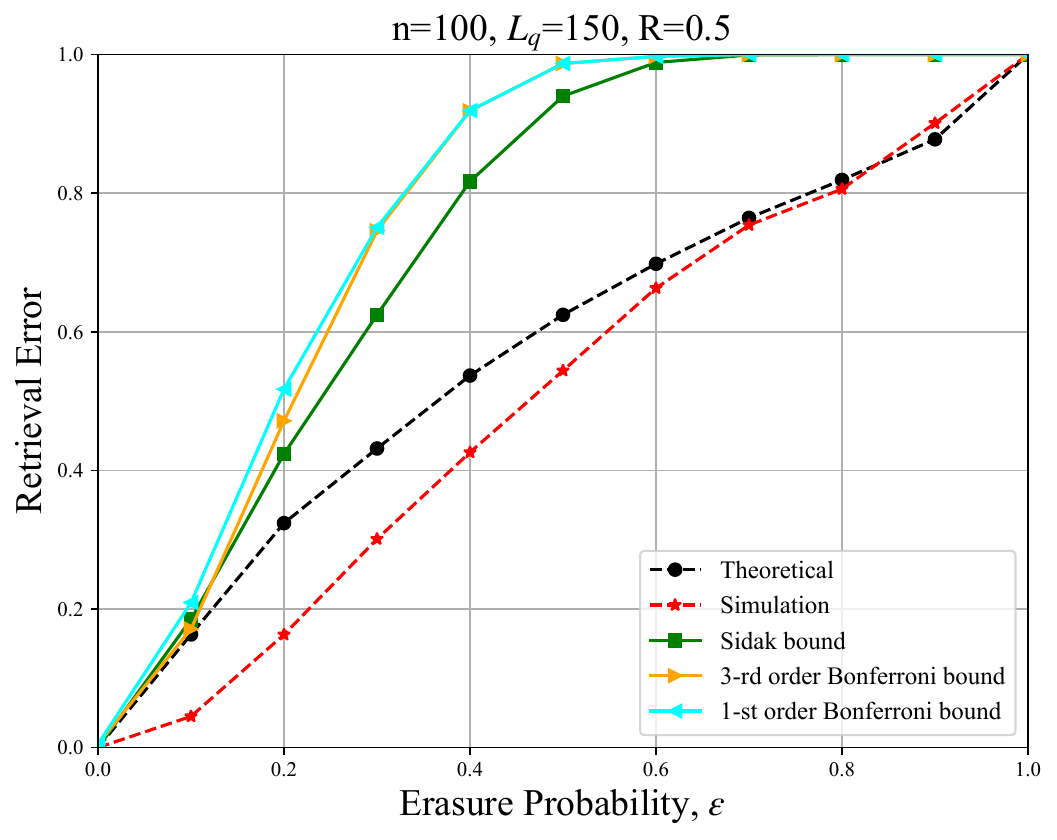}

\caption{Retrieval error probability versus erasure probability for $n=100$ documents, shown for different query lengths $L_q$ and data rates $R=1$ and $\frac{1}{2}$.}
\label{fig:error_synthetic_n_100}
\end{figure*}
\vspace{-10pt}
\subsection{Embedding-Based Retrieval Results (Data-Driven)}
In addition to the TF-IDF–based simulations, we evaluate a data-driven retrieval pipeline based on pretrained semantic embeddings. This experiment departs from the analytical model in Sec.~\ref{sec: error_pr_MVN} and is intended to assess whether semantic-aware repetition remains effective in a realistic embedding-based setting, rather than to validate the theoretical error bounds.
The empirical error probability is estimated over Monte Carlo trials, each consisting of semantic-importance estimation, optimized rate allocation, erasure realization, query reconstruction, and embedding-based retrieval. This data-driven evaluation complements the TF-IDF analysis by demonstrating the benefit of semantic-aware repetition under token-level erasures in modern embedding-based retrieval systems.

Fig.~\ref{fig:googleNQ_TFIDF} shows the retrieval error on the Google NQ corpus as a function of the erasure probability~$\epsilon$ under LLM-based retrieval. Performance is evaluated using Top-$K$ accuracy, where retrieval is considered correct if the ground-truth document appears among the top $K$ ranked results ($K=1,3,5$).
 As expected, the error increases monotonically with~$\epsilon$ and approaches one under severe erasures. Across all Top-$K$ metrics, the proposed importance-aware repetition scheme consistently outperforms the no-encoding baseline. 
Notably, even under compression ($R=2$), the semantic-aware scheme achieves lower error than the uncoded baseline ($R=1$), because redundancy is selectively assigned to the most informative query terms, whereas the baseline transmits each term only once with no protection against erasures. Increasing redundancy (smaller $R$) further reduces the error by increasing the probability that critical features survive the erasure channel. Allowing larger $K$ lowers the absolute error because the correct document is more likely to appear among multiple high-ranked candidates, while the relative performance ordering remains unchanged, indicating that semantic-aware protection improves the overall ranking quality rather than only the top result.
\begin{figure}
    \centering
    \includegraphics[width=0.37\textwidth,height=0.17\textheight]{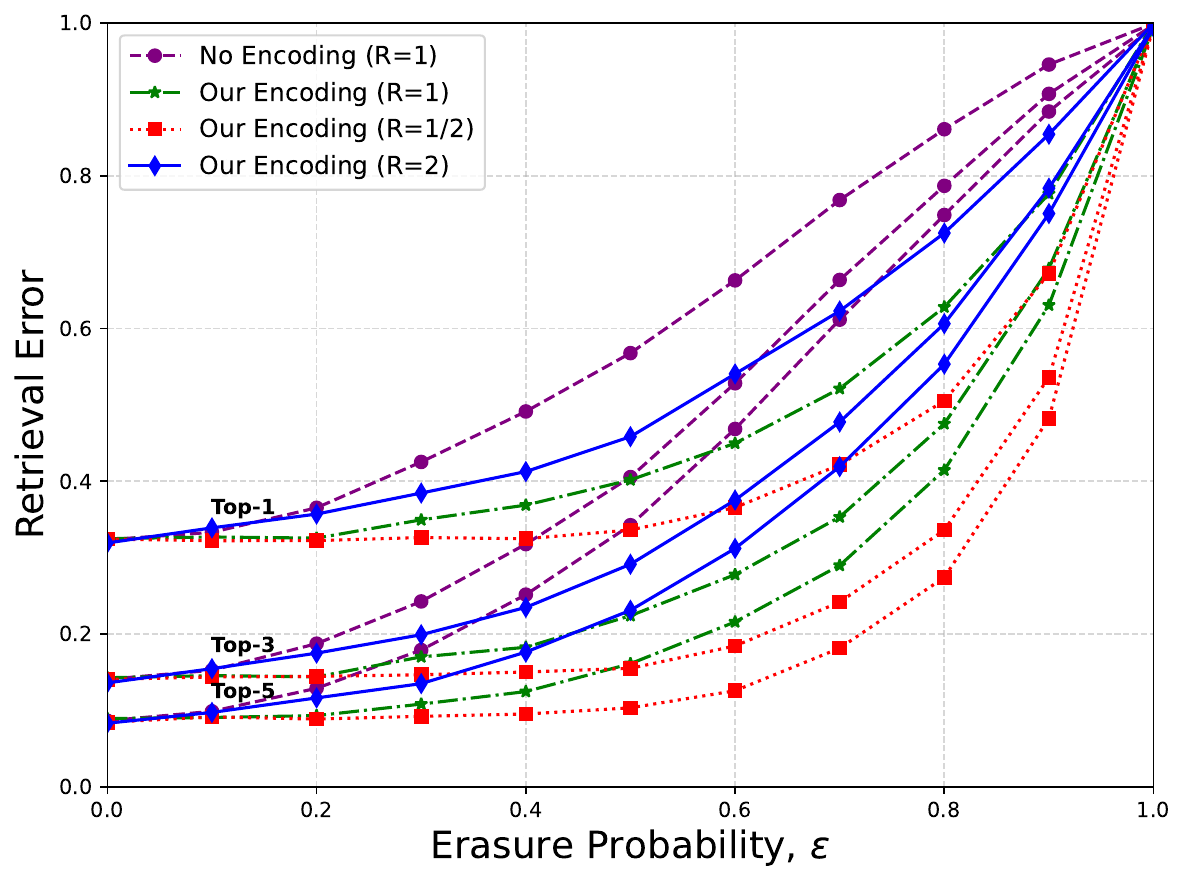}
    \caption{Google NQ ( intfloat/e5-large-v2) retrieval robustness under erasures: error probability as a function of $\epsilon$ for baseline and proposed rate-adaptive repetition schemes, evaluated for Top-$k$ retrieval.}
    \label{fig:googleNQ_TFIDF}
\end{figure}
\vspace{-6pt}
\section{Conclusion}
\label{sec: conclusion}
In conclusion, this paper establishes an information-theoretic framework for understanding remote document retrieval under token erasures and shows that importance-aware redundancy can significantly improve retrieval reliability. By deriving a multivariate Gaussian approximation for TF-IDF score margins and developing computable error bounds, the paper provides both theoretical insight and practical tools for analyzing retrieval performance. The numerical and data-driven results further show that protecting semantically important query components leads to more robust retrieval, not only in the analytical TF-IDF setting but also in modern embedding-based pipelines.
\appendices
\vspace{-5pt}
\section{Proof of Lemma~\ref{lem:lindeberg}}
\label{app:linbergcondition}
Condition on \(\mathbf v_q\), fix a unit vector \(\mathbf{u}\in\mathbb R^m\) and \(\eta>0\), and write
\[
a_i := \mathbf{u}^\top \boldsymbol\delta_i,
\quad
Z_i := \mathbf{u}^\top \mathbf Y_i = (e_i-p_i)a_i,
\quad i\in\mathcal S_q.
\]
where $ \boldsymbol\delta_i$ is defined in \eqref{eq:delta_i}. Then
\[
s_u^2
:=
\sum_{i\in\mathcal S_q}\mathrm{Var}(Z_i\mid \mathbf v_q)
=
\mathbf{u}^\top \boldsymbol\Sigma(\mathbf v_q)\,\mathbf{u}.
\]
If \(s_u^2=0\), then \(Z_i=0\) almost surely for every \(i\in\mathcal S_q\), and the directional Lindeberg ratio is \(0\) by convention. Hence it suffices to consider the case \(s_u^2>0\).

Let \(r_0\ge 1\) be a fixed finite integer, independent of $K_q$, and define
\[
G_{r_0}:=\{i\in \mathcal S_q:\ r_i\le r_0\},
\quad
B_{r_0}:= \mathcal S_q\setminus G_{r_0},
\]
where \(p_i=1-\epsilon^{r_i}\) and \(\epsilon\in(0,1)\) is fixed. 
Assume that the following technical regularity conditions, which are used to verify the directional Lindeberg condition, hold as \(K_q\to\infty\).

\begin{enumerate}
\item[(A1)] \emph{Singleton-dominated regime on the good set:}
there exist constants \(0<c_1\le c_2<\infty\), independent of \(K_q\), such that for all sufficiently large \(K_q\),
\[
\frac{c_1}{L_q}\le v_{q,i}\le \frac{c_2}{L_q}
\quad \text{for all } i\in G_{r_0},
\]
and \(|G_{r_0}|\to\infty\).

 (A1) captures a singleton-dominated sparse-query regime: in the TF-based model, normalized term-frequency weights for single occurrences are naturally of order \(1/L_q\), and (A1) requires a growing subset \(G_{r_0}\) of active coordinates to remain uniformly comparable to this scale.

\item[(A2)] \emph{Uniform boundedness of document differences and IDF weights:}
there exist constants \(C_d,C_\zeta<\infty\) such that for all \(i\in \mathcal S_q\) and all \(j\neq k\),
\[
|v_{d_k,i}-v_{d_j,i}|\le C_d,
\quad
\zeta_i \le C_\zeta.
\]
This assumption is a mild boundedness condition ensuring uniform control of the coordinates entering $\delta_i$. In standard TF-IDF settings, normalized document TF coordinates are bounded and IDF weights are finite for any finite corpus.
\item[(A3)] \emph{Directional non-degeneracy on the good set:}
for every fixed unit vector \(\mathbf{u}\in\mathbb{R}^m\), there exists a constant \(c_u>0\), independent of \(K_q\), such that for all sufficiently large \(K_q\),
\[
\sum_{i\in G_{r_0}} \bigl(\mathbf{u}^\top \boldsymbol\delta_i\bigr)^2
\ge
c_u \sum_{i\in G_{r_0}} v_{q,i}^2.
\]
(A3) ensures that, in every fixed projection direction \(\mathbf u\), the good-set vectors \(\boldsymbol\delta_i\) retain nontrivial projected energy. It rules out degenerate cases in which \(\mathbf u^\top \boldsymbol\delta_i\) becomes collectively too small, and is used to lower-bound the projected variance in the Lindeberg argument.

\item[(A4)] \emph{Tail \(\ell_2\)-negligibility:}
\[
\frac{\sum_{i\in B_{r_0}} v_{q,i}^2}
     {\sum_{i\in G_{r_0}} v_{q,i}^2}
\to 0.
\]
(A4) requires the \(\ell_2\)-mass of the large-repetition tail \(B_{r_0}\) to be negligible relative to that of the bounded-repetition set \(G_{r_0}\). Thus, the query energy is asymptotically dominated by the good set, and the bad-set contribution does not affect the Lindeberg verification.
\end{enumerate}

We show that
\[
\frac{1}{s_u^2}
\sum_{i\in \mathcal S_q}
\mathbb E\!\left[
Z_i^2\,\mathbf 1\{|Z_i|>\eta s_u\}
\,\middle|\, \mathbf v_q
\right]
\to 0,
\quad \text{as } K_q\to\infty,
\]
which is exactly \eqref{eq:lindeberg_directional}.

Since \(e_i\in\{0,1\}\) and \(p_i\in[0,1]\), we have \(|e_i-p_i|\le 1\) almost surely, and therefore
\begin{equation}
\label{eq:Zi_bound_appendix}
|Z_i|\le |a_i|
\quad \text{a.s.}
\end{equation}
Moreover, for each competitor \(j\neq k\),
$\delta_{j,i}=2\zeta_i^2 v_{q,i}(v_{d_k,i}-v_{d_j,i})$,
so by (A2),
$|\delta_{j,i}|
\le
2C_\zeta^2 C_d\, v_{q,i}$.
Hence
\[
\|\boldsymbol\delta_i\|_2^2
\le
m(2C_\zeta^2 C_d)^2 v_{q,i}^2.
\]
Since \(\|\mathbf{u}\|_2=1\), it follows that
\begin{equation}
\label{eq:ai_bound_appendix}
a_i^2=(\mathbf{u}^\top\boldsymbol\delta_i)^2
\le \|\boldsymbol\delta_i\|_2^2
\le K^2 v_{q,i}^2,
\quad
K:=2C_\zeta^2C_d\sqrt{m}.
\end{equation}

We now split the Lindeberg sum over the good and bad sets:
\[
\frac{1}{s_u^2}
\sum_{i\in \mathcal S_q}
\mathbb E\!\left[
Z_i^2\,\mathbf 1\{|Z_i|>\eta s_u\}
\,\middle|\, \mathbf v_q
\right]
=
T_G+T_B,
\]
where
\[
T_G
:=
\frac{1}{s_u^2}\sum_{i\in G_{r_0}}
\mathbb E\!\left[
Z_i^2\,\mathbf 1\{|Z_i|>\eta s_u\}
\,\middle|\, \mathbf v_q
\right],
\]
and
\[
T_B
:=
\frac{1}{s_u^2}\sum_{i\in B_{r_0}}
\mathbb E\!\left[
Z_i^2\,\mathbf 1\{|Z_i|>\eta s_u\}
\,\middle|\, \mathbf v_q
\right].
\]

\paragraph*{Contribution of the good set}
For \(i\in G_{r_0}\),
\[
p_i(1-p_i)\!=\!\epsilon^{r_i}(1-\epsilon^{r_i})
\!\ge\!
c_{\epsilon,r_0},
\quad
c_{\epsilon,r_0}\!\!:=\!\!\min_{1\le r\le r_0}\!\!\epsilon^r(1-\epsilon^r)\!>\!0.
\]
Therefore,
\[
s_u^2
=
\sum_{i\in \mathcal S_q} p_i(1-p_i)a_i^2
\ge
\sum_{i\in G_{r_0}} p_i(1-p_i)a_i^2
\ge
c_{\epsilon,r_0}\!\sum_{i\in G_{r_0}} a_i^2.
\]
Hence
\[
\max_{i\in G_{r_0}}\frac{a_i^2}{s_u^2}
\le
\frac{1}{c_{\epsilon,r_0}}
\frac{\max_{i\in G_{r_0}} a_i^2}{\sum_{j\in G_{r_0}} a_j^2}.
\]
Using \eqref{eq:ai_bound_appendix} and (A3),
\[
\frac{\max_{i\in G_{r_0}} a_i^2}{\sum_{j\in G_{r_0}} a_j^2}
\le
\frac{K^2 \max_{i\in G_{r_0}} v_{q,i}^2}
     {c_u \sum_{j\in G_{r_0}} v_{q,j}^2}.
\]
By (A1), for all sufficiently large \(K_q\),
\[
\max_{i\in G_{r_0}} v_{q,i}^2 \le \frac{c_2^2}{L_q^2},
\quad
\sum_{j\in G_{r_0}} v_{q,j}^2 \ge |G_{r_0}|\,\frac{c_1^2}{L_q^2}.
\]
Thus
\[
\frac{\max_{i\in G_{r_0}} a_i^2}{\sum_{j\in G_{r_0}} a_j^2}
\le
\frac{K^2 c_2^2}{c_u c_1^2}\cdot \frac{1}{|G_{r_0}|}
\to 0,
\]
since \(|G_{r_0}|\to\infty\). Consequently,
\[
\max_{i\in G_{r_0}}\frac{a_i^2}{s_u^2}\to 0.
\]
Therefore, for all sufficiently large \(K_q\),
\[
|a_i|\le \eta s_u
\quad \text{for all } i\in G_{r_0}.
\]
Combining this with \eqref{eq:Zi_bound_appendix}, we get
\[
\mathbf 1\{|Z_i|>\eta s_u\}=0
\quad \text{for all } i\in G_{r_0},
\]
for all sufficiently large \(K_q\). Hence \(T_G=0\) eventually.

\paragraph*{Contribution of the bad set}
Using \(Z_i^2\mathbf 1\{|Z_i|>\eta s_u\}\le Z_i^2\), we obtain
\[
T_B
\le
\frac{1}{s_u^2}\sum_{i\in B_{r_0}} \mathbb E[Z_i^2\mid \mathbf v_q]
=
\frac{\sum_{i\in B_{r_0}} p_i(1-p_i)a_i^2}{s_u^2}.
\]
Since \(p_i(1-p_i)\le 1/4\), \eqref{eq:ai_bound_appendix} yields
\[
\sum_{i\in B_{r_0}} p_i(1-p_i)a_i^2
\le
\frac14 \sum_{i\in B_{r_0}} a_i^2
\le
\frac{K^2}{4}\sum_{i\in B_{r_0}} v_{q,i}^2.
\]
For the denominator, restricting to \(G_{r_0}\) and using (A3) gives
\[
s_u^2
\ge
\sum_{i\in G_{r_0}} p_i(1-p_i)a_i^2
\ge
c_{\epsilon,r_0}\sum_{i\in G_{r_0}} a_i^2
\ge
c_{\epsilon,r_0}c_u \sum_{i\in G_{r_0}} v_{q,i}^2.
\]
Therefore,
\[
T_B
\le
\frac{K^2}{4c_{\epsilon,r_0}c_u}
\cdot
\frac{\sum_{i\in B_{r_0}} v_{q,i}^2}
     {\sum_{i\in G_{r_0}} v_{q,i}^2}
\to 0
\]
by (A4).

Combining the good- and bad-set bounds, we conclude that
\[
\frac{1}{s_u^2}
\sum_{i\in \mathcal S_q}
\mathbb E\!\left[
Z_i^2\,\mathbf 1\{|Z_i|>\eta s_u\}
\,\middle|\, \mathbf v_q
\right]
\to 0,
\quad \text{as } K_q\to\infty.
\]
This verifies the directional Lindeberg condition \eqref{eq:lindeberg_directional}.
\vspace{-5pt}
\section{Proof of Lemma~\ref{lem:cov_concentration}}
\label{app:covariance_positivity}

Fix two distinct competitors \(j\neq \ell\). Recall that
\begin{equation*}
\Sigma_{j\ell}(\mathbf v_q)
=
\sum_{i\in \mathcal S_q} p_i(1-p_i)\,\delta_{j,i}\delta_{\ell,i},
\end{equation*}
where
$\delta_{j,i} = 2\,\zeta_i^2 v_{q,i}\bigl(v_{d_k,i}-v_{d_j,i}\bigr)$.
Define
$a_i := 4\,p_i(1-p_i)\,\zeta_i^4 v_{q,i}^2$,
$W_i := (v_{d_k,i}-v_{d_j,i})(v_{d_k,i}-v_{d_\ell,i})$, and
$U_i := (v_{d_k,i}-v_{d_j,i})^2$, then
\begin{equation*}
\Sigma_{j\ell}(\mathbf v_q)=\sum_{i\in \mathcal S_q} a_i W_i,
\qquad
\Sigma_{jj}(\mathbf v_q)=\sum_{i\in \mathcal S_q} a_i U_i.
\end{equation*}

For each \(i\in\mathcal S_q\), the random variables \(v_{d_k,i}\), \(v_{d_j,i}\), and \(v_{d_\ell,i}\) are i.i.d.\ copies of a real random variable \(F\) with \(\mathbb E[F^2]<\infty\) and \(\mathrm{Var}(F)>0\), and the triples \((v_{d_k,i},v_{d_j,i},v_{d_\ell,i})\) are independent across \(i\). Let \(F_1,F_2,F_3\) be i.i.d.\ copies of \(F\). Then
\begin{equation*}
W_i \overset{d}= (F_1-F_2)(F_1-F_3),
\qquad
U_i \overset{d}= (F_1-F_2)^2,
\end{equation*}
so
\begin{equation*}
\mathbb E[W_i]
=
\mathbb E[(F_1-F_2)(F_1-F_3)]
=
\mathrm{Var}(F),
\end{equation*}
and
\begin{equation*}
\mathbb E[U_i]
=
\mathbb E[(F_1-F_2)^2]
=
2\,\mathrm{Var}(F).
\end{equation*}
In particular, \(W_i\) and \(U_i\) have finite first moments.

Now let
$A_{K_q}:=\sum_{i\in\mathcal S_q} a_i$.
By assumption,
\begin{equation*}
0<c_1 K_q \le A_{K_q}\le c_2 K_q,
\qquad
\sup_{i\in\mathcal S_q} a_i<\infty,
\end{equation*}
for all sufficiently large \(K_q\). Hence
\begin{equation*}
\frac{\max_{i\in\mathcal S_q} a_i}{A_{K_q}}
\le
\frac{\sup_{i\in\mathcal S_q} a_i}{c_1 K_q}
\xrightarrow[K_q\to\infty]{}0.
\end{equation*}
Therefore, by a weighted strong law of large numbers,
\begin{equation*}
\frac{\sum_{i\in S_q} a_i W_i}{A_{K_q}}
\xrightarrow[K_q\to\infty]{\mathrm{a.s.}}
\mathbb E[W_i]
=
\mathrm{Var}(F).
\end{equation*}
Likewise,
\begin{equation*}
\frac{\sum_{i\in S_q} a_i U_i}{A_{K_q}}
\xrightarrow[K_q\to\infty]{\mathrm{a.s.}}
\mathbb E[U_i]
=
2\,\mathrm{Var}(F).
\end{equation*}
Hence
\begin{equation*}
\begin{aligned}
\rho_{j\ell}(\mathbf v_q)
=
\frac{\Sigma_{j\ell}(\mathbf v_q)}
{\sqrt{\Sigma_{jj}(\mathbf v_q)\Sigma_{\ell\ell}(\mathbf v_q)}}
&\xrightarrow[K_q\to\infty]{\mathrm{a.s.}}\\
&\frac{\mathrm{Var}(F)}
{\sqrt{(2\,\mathrm{Var}(F))(2\,\mathrm{Var}(F))}}
=
\frac12.
\end{aligned}
\end{equation*}
In particular, \(\rho_{j\ell}(\mathbf v_q)>0\) for all sufficiently large \(K_q\), almost surely.
\hfill\IEEEQED




 




\bibliographystyle{IEEEtran}
\bibliography{Refs}

@article{ibrihich2022review,
  title={A Review on recent research in information retrieval},
  author={Ibrihich, S and Oussous, A and Ibrihich, O and Esghir, M},
  journal={Procedia Computer Science},
  volume={201},
  pages={777--782},
  year={2022},
  publisher={Elsevier}
}

@article{piantadosi2014zipf,
  title={Zipf’s word frequency law in natural language: A critical review and future directions},
  author={Piantadosi, Steven T},
  journal={Psychonomic Bulletin \& Review},
  volume={21},
  pages={1112--1130},
  year={2014},
  publisher={Springer}
}

@inproceedings{fox1989stop,
  title={A stop list for general text},
  author={Fox, Christopher},
  booktitle={Acm Sigir Forum},
  volume={24},
  number={1-2},
  pages={19--21},
  year={1989},
  organization={ACM New York, NY, USA}
}

@article{dimitrakis2020survey,
  title={A survey on question answering systems over linked data and documents},
  author={Dimitrakis, Eleftherios and Sgontzos, Konstantinos and Tzitzikas, Yannis},
  journal={Journal of Intelligent Information Systems},
  volume={55},
  number={2},
  pages={233--259},
  year={2020},
  publisher={Springer}
}

@article{zhu2021retrieving,
  title={Retrieving and reading: {A} comprehensive survey on open-domain question answering},
  author={Zhu, Fengbin and Lei, Wenqiang and Wang, Chao and Zheng, Jianming and Poria, Soujanya and Chua, Tat-Seng},
  journal={arXiv preprint arXiv:2101.00774},
  year={2021}
}

@inproceedings{mishra2015analysis,
  title={Analysis of {TF-IDF} model and its variant for document retrieval},
  author={Mishra, Apra and Vishwakarma, Santosh},
  booktitle={2015 
             international Conference on Computational Intelligence and Communication Networks (CICN)},
  pages={772--776},
  year={2015},
  organization={}
}

@article{gao2023retrieval,
  title={Retrieval-augmented generation for large language models: A survey},
  author={Gao, Yunfan and Xiong, Yun and Gao, Xinyu and Jia, Kangxiang and Pan, Jinliu and Bi, Yuxi and Dai, Yi and Sun, Jiawei and Wang, Haofen},
  journal={arXiv preprint arXiv:2312.10997},
  year={2023}
}

@inproceedings{salemi2024evaluating,
  title={Evaluating retrieval quality in retrieval-augmented generation},
  author={Salemi, Alireza and Zamani, Hamed},
  booktitle={Proceedings of the 47th International ACM SIGIR Conference on Research and Development in Information Retrieval},
  pages={2395--2400},
  year={2024}
}

@article{li2024enhancing,
  title={Enhancing llm factual accuracy with rag to counter hallucinations: A case study on domain-specific queries in private knowledge-bases},
  author={Li, Jiarui and Yuan, Ye and Zhang, Zehua},
  journal={arXiv preprint arXiv:2403.10446},
  year={2024}
}

@book{buttcher2016information,
  title={Information retrieval: {I}mplementing and evaluating search engines},
  author={Buttcher, Stefan and Clarke, Charles LA and Cormack, Gordon V},
  year={2016},
  publisher={{MIT} Press}
}

@article{siriwardhana2023improving,
  title={Improving the domain adaptation of retrieval augmented generation ({RAG}) models for open domain question answering},
  author={Siriwardhana, Shamane and Weerasekera, Rivindu and Wen, Elliott and Kaluarachchi, Tharindu and Rana, Rajib and Nanayakkara, Suranga},
  journal={Transactions of the Association for Computational Linguistics},
  volume={11},
  pages={1--17},
  year={2023},
  publisher={MIT Press One Broadway, 12th Floor, Cambridge, Massachusetts 02142, USA~…}
}

@article{lu2023semantics,
  title={Semantics-empowered communications: A tutorial-cum-survey},
  author={Lu, Zhilin and Li, Rongpeng and Lu, Kun and Chen, Xianfu and Hossain, Ekram and Zhao, Zhifeng and Zhang, Honggang},
  journal={IEEE Communications Surveys \& Tutorials},
  year={2023},
  publisher={IEEE}
}

@article{luo2022semantic,
  title={Semantic communications: Overview, open issues, and future research directions},
  author={Luo, Xuewen and Chen, Hsiao-Hwa and Guo, Qing},
  journal={IEEE Wireless Communications},
  volume={29},
  number={1},
  pages={210--219},
  year={2022},
  publisher={IEEE}
}

@article{shi2021semantic,
  title={From semantic communication to semantic-aware networking: Model, architecture, and open problems},
  author={Shi, Guangming and Xiao, Yong and Li, Yingyu and Xie, Xuemei},
  journal={IEEE Communications Magazine},
  volume={59},
  number={8},
  pages={44--50},
  year={2021},
  publisher={IEEE}
}

@article{lan2021semantic,
  title={What is semantic communication? {A} view on conveying meaning in the era of machine intelligence},
  author={Lan, Qiao and Wen, Dingzhu and Zhang, Zezhong and Zeng, Qunsong and Chen, Xu and Popovski, Petar and Huang, Kaibin},
  journal={Journal of Communications and Information Networks},
  volume={6},
  number={4},
  pages={336--371},
  year={2021},
  publisher={PTP}
}

@article{yang2022semantic,
  title={Semantic communications for future internet: Fundamentals, applications, and challenges},
  author={Yang, Wanting and Du, Hongyang and Liew, Zi Qin and Lim, Wei Yang Bryan and Xiong, Zehui and Niyato, Dusit and Chi, Xuefen and Shen, Xuemin and Miao, Chunyan},
  journal={IEEE Communications Surveys \& Tutorials},
  volume={25},
  number={1},
  pages={213--250},
  year={2022},
  publisher={IEEE}
}

@article{chaccour2024less,
  title={Less data, more knowledge: Building next generation semantic communication networks},
  author={Chaccour, Christina and Saad, Walid and Debbah, M{\'e}rouane and Han, Zhu and Poor, H Vincent},
  journal={IEEE Communications Surveys \& Tutorials},
  year={2024},
  publisher={IEEE}
}

@article{kwiatkowski2019natural,
  title={Natural questions: a benchmark for question answering research},
  author={Kwiatkowski, Tom and Palomaki, Jennimaria and Redfield, Olivia and Collins, Michael and Parikh, Ankur and Alberti, Chris and Epstein, Danielle and Polosukhin, Illia and Devlin, Jacob and Lee, Kenton and others},
  journal={Transactions of the Association for Computational Linguistics},
  volume={7},
  pages={453--466},
  year={2019},
  publisher={MIT Press One Rogers Street, Cambridge, MA 02142-1209, USA journals-info~…}
}

@article{dubey2024llama,
  title={The llama 3 herd of models},
  author={Dubey, Abhimanyu and Jauhri, Abhinav and Pandey, Abhinav and Kadian, Abhishek and Al-Dahle, Ahmad and Letman, Aiesha and Mathur, Akhil and Schelten, Alan and Yang, Amy and Fan, Angela and others},
  journal={arXiv e-prints},
  pages={arXiv--2407},
  year={2024}
}

@article{guo2023semantic,
  title={Semantic importance-aware communications using pre-trained language models},
  author={Guo, Shuaishuai and Wang, Yanhu and Li, Shujing and Saeed, Nasir},
  journal={IEEE Communications Letters},
  volume={27},
  number={9},
  pages={2328--2332},
  year={2023},
  publisher={IEEE}
}

@inproceedings{liu2025text,
  title={Text-Guided Token Communication for Wireless Image Transmission},
  author={Liu, Bole and Qiao, Li and Wang, Ye and Gao, Zhen and Ma, Yu and Ying, Keke and Qin, Tong},
  booktitle={2025 IEEE/CIC International Conference on Communications in China (ICCC)},
  pages={1--6},
  year={2025},
  organization={}
}

@article{qiao2025token,
  title={Token communications: A large model-driven framework for cross-modal context-aware semantic communications},
  author={Qiao, Li and Mashhadi, Mahdi Boloursaz and Gao, Zhen and Tafazolli, Rahim and Bennis, Mehdi and Niyato, Dusit},
  journal={IEEE Wireless Communications},
  volume={32},
  number={5},
  pages={80--88},
  year={2025},
  publisher={IEEE}
}

@article{vsidak1967rectangular,
  title={Rectangular confidence regions for the means of multivariate normal distributions},
  author={{\v{S}}id{\'a}k, Zbyn{\v{e}}k},
  journal={Journal of the American Statistical Association},
  volume={62},
  number={318},
  pages={626--633},
  year={1967},
  publisher={Taylor \& Francis}
}

@book{feller1991introduction,
  title={An Introduction to Probability Theory and its Applications, Volume 2},
  author={Feller, William},
  volume={2},
  year={1991},
  publisher={John Wiley \& Sons}
}

@book{van2000asymptotic,
  title={Asymptotic statistics},
  author={Van der Vaart, Aad W},
  volume={3},
  year={2000},
  publisher={Cambridge University Press}
}

@inproceedings{ghasvarianjahromi2025context,
  title={Context-Aware Search and Retrieval Over Erasure Channels},
  author={Ghasvarianjahromi, Sara and Yakimenka, Yauhen and Kliewer, J{\"o}rg},
  booktitle={2025 IEEE Information Theory Workshop (ITW)},
  pages={821--826},
  year={2025},
  organization={}
}

@article{knollmeyer2025hybrid,
  title={Hybrid Retrieval for Retrieval Augmented Generation in the German Language Production Domain},
  author={Knollmeyer, Simon and Pfaff, Sebastian and Akmal, Muhammad Uzair and Koval, Leonid and Asif, Saara and Mathias, Selvine G and Gro{\ss}mann, Daniel},
  journal={Journal of Advances in Information Technology},
  volume={16},
  number={6},
  year={2025}
}

@inproceedings{xian2024vector,
  title={Vector search with OpenAI embeddings: Lucene is all you need},
  author={Xian, Jasper and Teofili, Tommaso and Pradeep, Ronak and Lin, Jimmy},
  booktitle={Proceedings of the 17th ACM International Conference on Web Search and Data Mining},
  pages={1090--1093},
  year={2024}
}

@inproceedings{lin2025gosling,
  title={Gosling Grows Up: Retrieval with Learned Dense and Sparse Representations Using Anserini},
  author={Lin, Jimmy and Chen, Arthur Haonan and Lassance, Carlos and Ma, Xueguang and Pradeep, Ronak and Teofili, Tommaso and Xian, Jasper and Yang, Jheng-Hong and Zhong, Brayden and Zhong, Vincent},
  booktitle={Proceedings of the 48th International ACM SIGIR Conference on Research and Development in Information Retrieval},
  pages={3223--3233},
  year={2025}
}

@article{wang2025balancing,
  title={Balancing the Blend: An Experimental Analysis of Trade-offs in Hybrid Search},
  author={Wang, Mengzhao and Tan, Boyu and Gao, Yunjun and Jin, Hai and Zhang, Yingfeng and Ke, Xiangyu and Xu, Xiaoliang and Zhu, Yifan},
  journal={arXiv preprint arXiv:2508.01405},
  year={2025}
}

@inproceedings{reimers2019sentence,
  title={Sentence-bert: Sentence embeddings using siamese bert-networks},
  author={Reimers, Nils and Gurevych, Iryna},
  booktitle={Proceedings of the 2019 Conference on Empirical Methods in Natural Language Processing and the 9th International Joint Conference on Natural Language Processing (EMNLP-IJCNLP)},
  pages={3982--3992},
  year={2019}
}

@article{mikolov2013efficient,
  title={Efficient estimation of word representations in vector space},
  author={Mikolov, Tomas and Chen, Kai and Corrado, Greg and Dean, Jeffrey},
  journal={arXiv preprint arXiv:1301.3781},
  year={2013}
}

@book{bellman1957dynamic,
  author    = {Richard Bellman},
  title     = {Dynamic Programming},
  publisher = {Princeton University Press},
  year      = {1957}
}

@book{billingsley2013convergence,
  title={Convergence of Probability Measures},
  author={Billingsley, Patrick},
  year={2013},
  publisher={John Wiley \& Sons}
}

@article{mcinnes2018umap,
  title={{UMAP}: Uniform manifold approximation and projection for dimension reduction},
  author={McInnes, Leland and Healy, John and Melville, James},
  journal={arXiv preprint arXiv:1802.03426},
  year={2018}
}
\end{document}